\documentclass[a4paper,12pt]{amsart}
\usepackage{amsmath,amssymb,amsthm}
\usepackage{amscd}
\usepackage{mathrsfs}
\usepackage[usenames,dvipsnames]{color}
\usepackage{graphicx}
\usepackage[all]{xy}
\usepackage{eucal}
\usepackage{extarrows} 
\usepackage{tabularx}
\usepackage{tikz} 
\usepackage{xparse} 
\usepackage{colordvi}
\usepackage{multicol}
\usepackage{enumitem}
\usepackage[normalem]{ulem}
\usepackage{epsfig}
 \usepackage{ifpdf}
  \ifpdf
   \usepackage[colorlinks,final,backref=page,hyperindex]{hyperref}
  \else
   \usepackage[colorlinks,final,backref=page,hyperindex,hypertex]{hyperref}
  \fi

\usepackage{xspace}

\usetikzlibrary{arrows}  
\usetikzlibrary{decorations.pathmorphing}  

\begin{document}

\newcommand\cutoffint{\mathop{-\hskip -4mm\int}\limits}
\newcommand\cutoffsum{\mathop{-\hskip -4mm\sum}\limits}
\newcommand\cutoffzeta{-\hskip -1.7mm\zeta} 
\newcommand{\goth}[1]{\ensuremath{\mathfrak{#1}}}
\newcommand{\bbox}{\normalsize {}%
        \nolinebreak \hfill $\blacksquare$ \medbreak \par}
\newcommand{\simall}[2]{\underset{#1\rightarrow#2}{\sim}}

\newtheorem{theorem}{Theorem}[section]
\newtheorem{prop}[theorem]{Proposition}
\newtheorem{lemdefn}[theorem]{Lemma-Definition}

\theoremstyle{definition}
\newtheorem{rk}[theorem]{Remark}
\theoremstyle{definition}
\newtheorem{defn}[theorem]{Definition}
\newtheorem{propdefn}[theorem]{Proposition-Definition}
\newtheorem{lem}[theorem]{Lemma}
\newtheorem{thm}[theorem]{Theorem}
\newtheorem{coro}[theorem]{Corollary}
\newtheorem{ex}[theorem]{Example}
\newtheorem{claim}[theorem]{Claim}
\newtheorem{coex}[theorem]{Counterexample}
\renewcommand{\theenumi}{{\it\roman{enumi}}}
\renewcommand{\theenumii}{\alpha{enumii}}

\newenvironment{thmenumerate}{\leavevmode\begin{enumerate}[leftmargin=1.5em]}{\end{enumerate}}

\newcommand{\delete}[1]{{}}

\newcommand{\nc}{\newcommand}

\nc{\mlabel}[1]{\label{#1}}  
\nc{\mcite}[1]{\cite{#1}}  
\nc{\mref}[1]{\ref{#1}}  
\nc{\mbibitem}[1]{\bibitem{#1}} 

\delete{
\nc{\mlabel}[1]{\label{#1}  
{\hfill \hspace{1cm}{\small{{\ }\hfill(#1)}}}}
\nc{\mcite}[1]{\cite{#1}{\small{{{\ }(#1)}}}}  
\nc{\mref}[1]{\ref{#1}{{{{\ }(#1)}}}}  
\nc{\mbibitem}[1]{\bibitem[\bf #1]{#1}} 
}



\nc{\ola}[1]{\stackrel{#1}{\longrightarrow}}
\nc{\mtop}{\top\hspace{-1mm}}
\nc{\mrm}[1]{{\rm #1}}
\nc{\depth}{{\mrm d}}
\nc{\id}{\mrm{id}}
\nc{\Id}{\mathrm{Id}}
\nc{\mapped}{operated\xspace}
\nc{\Mapped}{Operated\xspace}
\newcommand{\redtext}[1]{{\textcolor{red}{#1}}}
\newcommand{\Hol}{\text{Hol}}
\newcommand{\Mer}{\text{Mer}}
\newcommand{\lin}{\text{lin}}
\nc{\ot}{\otimes}
\nc{\Hom}{\mathrm{Hom}}
\nc{\CS}{\mathcal{CS}}
\nc{\bfk}{\mathbf{K}}
\nc{\lwords}{\calw}
\nc{\ltrees}{\calf}
\nc{\lpltrees}{\calp}
\nc{\Map}{\mathrm{Map}}
\nc{\rep}{\beta}
\nc{\free}[1]{\bar{#1}}
\nc{\OS}{\mathbf{OS}}
\nc{\OM}{\mathbf{OM}}
\nc{\OA}{\mathbf{OA}}
\nc{\based}{based\xspace}
\nc{\tforall}{\text{ for all }}
\nc{\hwp}{\widehat{P}^\calw}
\nc{\sha}{{\mbox{\cyr X}}}
\font\cyr=wncyr10 \font\cyrs=wncyr7
\nc{\Mor}{\mathrm{Mor}}
\def\lc{\lfloor}
\def\rc{\rfloor}
\nc{\oF}{{\overline{F}}}
\nc{\mge}{_{bu}\!\!\!\!{}}
\newcommand{\bottop}{\top\hspace{-0.8em}\bot}
\nc{\supp}{\mathrm{Dep}}
\nc{\orth}{orthogonal\xspace}
\newcommand{\W}{\mathbb{W}}
\newcommand{\R}{\mathbb{R}}
\newcommand{\bbR}{\mathbb{R}}
\newcommand{\C}{\mathbb{C}}
\newcommand{\K}{\mathbb{K}}
\newcommand{\Z}{\mathbb{Z}}
\newcommand{\Q}{\mathbb{Q}}
\newcommand{\bbB}{\mathbb{B}}
\newcommand{\N}{\mathbb{N}}
\newcommand{\F}{\mathbf{F}}
\newcommand{\T}{\mathbf{T}}
\newcommand{\bbG}{\mathbb{G}}
\newcommand{\U}{\mathbb{U}}
\newcommand{\loc}{locality\xspace}
\newcommand{\Loc}{Locality\xspace}
\newcommand {\frakc}{{\mathfrak {c}}}
\newcommand {\frakd}{{\mathfrak {d}}}
\newcommand {\fraku}{{\mathfrak {u}}}
\newcommand {\fraks}{{\mathfrak {s}}}
\newcommand{\frakS}{S}
\newcommand {\bbf}{\ltrees}
\newcommand {\bbg}{{\mathbb{G}}}
\newcommand {\bbp}{\lpltrees}
\newcommand {\bbw}{{\mathbb{W}}}
\newcommand {\cala}{{\mathcal {A}}}
\newcommand {\calb}{\mathcal {B}}
\newcommand {\calc}{{\mathcal {C}}}
\newcommand {\cald}{{\mathcal {D}}}
\newcommand {\cale}{{\mathcal {E}}}
\newcommand {\calf}{{\mathcal {F}}}
\newcommand {\calg}{{\mathcal {G}}}
\newcommand {\calh}{\mathcal{H}}
\newcommand {\cali}{\mathcal{I}}
\newcommand {\call}{{\mathcal {L}}}
\newcommand {\calm}{{\mathcal {M}}}
\newcommand {\calp}{{\mathcal {P}}}
\newcommand {\calr}{{\mathcal {R}}}
\newcommand {\cals}{{\mathcal {S}}}
\newcommand {\calt}{{\mathcal {T}}}
\newcommand {\calv}{{\mathcal {V}}}
\newcommand {\calw}{{\mathcal {W}}}
\nc{\vep}{\varepsilon}
\def \e {{\epsilon}}
\newcommand{\sy  }[1]{{\color{purple}  #1}} 
\newcommand{\cy}[1]{{\color{cyan}  #1}}
\newcommand{\zb }[1]{{\color{blue}  #1}}
\newcommand{\li}[1]{{\color{red} #1}}
\newcommand{\lir}[1]{{\it\color{red} (Li: #1)}}


\newcommand{\tddeux}[2]{\begin{picture}(12,5)(0,-1)
\put(3,0){\circle*{2}}
\put(3,0){\line(0,1){5}}
\put(3,5){\circle*{2}}
\put(3,-2){\tiny #1}
\put(3,4){\tiny #2}
\end{picture}}

\newcommand{\tdtroisun}[3]{\begin{picture}(20,12)(-5,-1)
\put(3,0){\circle*{2}}
\put(-0.65,0){$\vee$}
\put(6,7){\circle*{2}}
\put(0,7){\circle*{2}}
\put(5,-2){\tiny #1}
\put(6,5){\tiny #2}
\put(-5,8){\tiny #3}
\end{picture}}

\def\ta1{{\scalebox{0.25}{ 
\begin{picture}(12,12)(38,-38)
\SetWidth{0.5} \SetColor{Black} \Vertex(45,-33){5.66}
\end{picture}}}}

\def\tb2{{\scalebox{0.25}{ 
\begin{picture}(12,42)(38,-38)
\SetWidth{0.5} \Vertex(45,-3){5.66}
\SetWidth{1.0} \Line(45,-3)(45,-33) \SetWidth{0.5}
\Vertex(45,-33){5.66}
\end{picture}}}}

\def\tc3{{\scalebox{0.25}{ 
\begin{picture}(12,72)(38,-38)
\SetWidth{0.5} \SetColor{Black} \Vertex(45,27){5.66}
\SetWidth{1.0} \Line(45,27)(45,-3) \SetWidth{0.5}
\Vertex(45,-33){5.66} \SetWidth{1.0} \Line(45,-3)(45,-33)
\SetWidth{0.5} \Vertex(45,-3){5.66}
\end{picture}}}}

\def\td31{{\scalebox{0.25}{ 
\begin{picture}(42,42)(23,-38)
\SetWidth{0.5} \SetColor{Black} \Vertex(45,-3){5.66}
\Vertex(30,-33){5.66} \Vertex(60,-33){5.66} \SetWidth{1.0}
\Line(45,-3)(30,-33) \Line(60,-33)(45,-3)
\end{picture}}}}

\def\xtd31{{\scalebox{0.35}{ 
\begin{picture}(70,42)(13,-35)
\SetWidth{0.5} \SetColor{Black} \Vertex(45,-3){5.66}
\Vertex(30,-33){5.66} \Vertex(60,-33){5.66} \SetWidth{1.0}
\Line(45,-3)(30,-33) \Line(60,-33)(45,-3)
\put(38,-38){\em \huge x}
\end{picture}}}}

\def\ytd31{{\scalebox{0.35}{ 
\begin{picture}(70,42)(13,-35)
\SetWidth{0.5} \SetColor{Black} \Vertex(45,-3){5.66}
\Vertex(30,-33){5.66} \Vertex(60,-33){5.66} \SetWidth{1.0}
\Line(45,-3)(30,-33) \Line(60,-33)(45,-3)
\put(38,-38){\em \huge y}
\end{picture}}}}

\def\xldec41r{{\scalebox{0.35}{ 
\begin{picture}(70,42)(13,-45)
\SetColor{Black}
\SetWidth{0.5} \Vertex(45,-3){5.66}
\Vertex(30,-33){5.66} \Vertex(60,-33){5.66}
\Vertex(60,-63){5.66}
\SetWidth{1.0}
\Line(45,-3)(30,-33) \Line(60,-33)(45,-3)
\Line(60,-33)(60,-63)
\put(38,-38){\em \huge x}

\end{picture}}}}

\def\xyrlong{{\scalebox{0.35}{ 
\begin{picture}(70,72)(13,-48)
\SetColor{Black}
\SetWidth{0.5} \Vertex(45,-3){5.66}
\Vertex(30,-33){5.66} \Vertex(60,-33){5.66} \SetWidth{1.0}
\Line(45,-3)(30,-33) \Line(60,-33)(45,-3)
\put(38,-38){\em\huge x}
\SetWidth{0.5}
\Vertex(45,-63){5.66} \Vertex(75,-63){5.66} \SetWidth{1.0}
\Line(60,-33)(45,-63) \Line(60,-33)(75,-63)
\put(55,-63){\em\huge y}
\end{picture}}}}

\def\xyllong{{\scalebox{0.35}{ 
\begin{picture}(70,72)(13,-48)
\SetColor{Black}
\SetWidth{0.5} \Vertex(45,-3){5.66}
\Vertex(30,-33){5.66} \Vertex(60,-33){5.66} \SetWidth{1.0}
\Line(45,-3)(30,-33) \Line(60,-33)(45,-3)
\put(40,-33){\em\huge y}
\SetWidth{0.5}
\Vertex(15,-63){5.66} \Vertex(45,-63){5.66} \SetWidth{1.0}
\Line(30,-33)(15,-63) \Line(30,-33)(45,-63)
\put(25,-63){\em\huge x}
\end{picture}}}}

\def\xyldec43{{\scalebox{0.35}{ 
\begin{picture}(70,62)(13,-25)
\SetColor{Black}
\SetWidth{0.5} \Vertex(45,-3){5.66}
\Vertex(15,-33){5.66} \Vertex(45,-38){5.66}
\Vertex(75,-33){5.66}
\SetWidth{1.0}
\Line(45,-3)(15,-33) \Line(45,-3)(45,-38)
\Line(45,-3)(74,-33)
\put(25,-33){\em\huge x}
\put(50,-33){\em\huge y}
\end{picture}}}}

\def\te4{{\scalebox{0.25}{ 
\begin{picture}(12,102)(38,-8)
\SetWidth{0.5} \SetColor{Black} \Vertex(45,57){5.66}
\Vertex(45,-3){5.66} \Vertex(45,27){5.66} \Vertex(45,87){5.66}
\SetWidth{1.0} \Line(45,57)(45,27) \Line(45,-3)(45,27)
\Line(45,57)(45,87)
\end{picture}}}}

\def\tf41{{\scalebox{0.25}{ 
\begin{picture}(42,72)(38,-8)
\SetWidth{0.5} \SetColor{Black} \Vertex(45,27){5.66}
\Vertex(45,-3){5.66} \SetWidth{1.0} \Line(45,27)(45,-3)
\SetWidth{0.5} \Vertex(60,57){5.66} \SetWidth{1.0}
\Line(45,27)(60,57) \SetWidth{0.5} \Vertex(75,27){5.66}
\SetWidth{1.0} \Line(75,27)(60,57)
\end{picture}}}}

\def\tg42{{\scalebox{0.25}{ 
\begin{picture}(42,72)(8,-8)
\SetWidth{0.5} \SetColor{Black} \Vertex(45,27){5.66}
\Vertex(45,-3){5.66} \SetWidth{1.0} \Line(45,27)(45,-3)
\SetWidth{0.5} \Vertex(15,27){5.66} \Vertex(30,57){5.66}
\SetWidth{1.0} \Line(15,27)(30,57) \Line(45,27)(30,57)
\end{picture}}}}

\def\th43{{\scalebox{0.25}{ 
\begin{picture}(42,42)(8,-8)
\SetWidth{0.5} \SetColor{Black} \Vertex(45,-3){5.66}
\Vertex(15,-3){5.66} \Vertex(30,27){5.66} \SetWidth{1.0}
\Line(15,-3)(30,27) \Line(45,-3)(30,27) \Line(30,27)(30,-3)
\SetWidth{0.5} \Vertex(30,-3){5.66}
\end{picture}}}}

\def\thII43{{\scalebox{0.25}{ 
\begin{picture}(72,57) (68,-128)
    \SetWidth{0.5}
    \SetColor{Black}
    \Vertex(105,-78){5.66}
    \SetWidth{1.5}
    \Line(105,-78)(75,-123)
    \Line(105,-78)(105,-123)
    \Line(105,-78)(135,-123)
    \SetWidth{0.5}
    \Vertex(75,-123){5.66}
    \Vertex(105,-123){5.66}
    \Vertex(135,-123){5.66}
  \end{picture}
  }}}

\def\thj44{{\scalebox{0.25}{ 
\begin{picture}(42,72)(8,-8)
\SetWidth{0.5} \SetColor{Black} \Vertex(30,57){5.66}
\SetWidth{1.0} \Line(30,57)(30,27) \SetWidth{0.5}
\Vertex(30,27){5.66} \SetWidth{1.0} \Line(45,-3)(30,27)
\SetWidth{0.5} \Vertex(45,-3){5.66} \Vertex(15,-3){5.66}
\SetWidth{1.0} \Line(15,-3)(30,27)
\end{picture}}}}

\def\xthj44{{\scalebox{0.35}{ 
\begin{picture}(42,72)(8,-8)
\SetWidth{0.5} \SetColor{Black} \Vertex(30,57){5.66}
\SetWidth{1.0} \Line(30,57)(30,27) \SetWidth{0.5}
\Vertex(30,27){5.66} \SetWidth{1.0} \Line(45,-3)(30,27)
\SetWidth{0.5} \Vertex(45,-3){5.66} \Vertex(15,-3){5.66}
\SetWidth{1.0} \Line(15,-3)(30,27)
\put(25,-3){\em\huge x}
\end{picture}}}}

\def\ti5{{\scalebox{0.25}{ 
\begin{picture}(12,132)(23,-8)
\SetWidth{0.5} \SetColor{Black} \Vertex(30,117){5.66}
\SetWidth{1.0} \Line(30,117)(30,87) \SetWidth{0.5}
\Vertex(30,87){5.66} \Vertex(30,57){5.66} \Vertex(30,27){5.66}
\Vertex(30,-3){5.66} \SetWidth{1.0} \Line(30,-3)(30,27)
\Line(30,27)(30,57) \Line(30,87)(30,57)
\end{picture}}}}

\def\tj51{{\scalebox{0.25}{ 
\begin{picture}(42,102)(53,-38)
\SetWidth{0.5} \SetColor{Black} \Vertex(61,27){4.24}
\SetWidth{1.0} \Line(75,57)(90,27) \Line(60,27)(75,57)
\SetWidth{0.5} \Vertex(90,-3){5.66} \Vertex(60,27){5.66}
\Vertex(75,57){5.66} \Vertex(90,-33){5.66} \SetWidth{1.0}
\Line(90,-33)(90,-3) \Line(90,-3)(90,27) \SetWidth{0.5}
\Vertex(90,27){5.66}
\end{picture}}}}

\def\tk52{{\scalebox{0.25}{ 
\begin{picture}(42,102)(23,-8)
\SetWidth{0.5} \SetColor{Black} \Vertex(60,57){5.66}
\Vertex(45,87){5.66} \SetWidth{1.0} \Line(45,87)(60,57)
\SetWidth{0.5} \Vertex(30,57){5.66} \SetWidth{1.0}
\Line(30,57)(45,87) \SetWidth{0.5} \Vertex(30,-3){5.66}
\SetWidth{1.0} \Line(30,-3)(30,27) \SetWidth{0.5}
\Vertex(30,27){5.66} \SetWidth{1.0} \Line(30,57)(30,27)
\end{picture}}}}

\def\tl53{{\scalebox{0.25}{ 
\begin{picture}(42,102)(8,-8)
\SetWidth{0.5} \SetColor{Black} \Vertex(30,57){5.66}
\Vertex(30,27){5.66} \SetWidth{1.0} \Line(30,57)(30,27)
\SetWidth{0.5} \Vertex(30,87){5.66} \SetWidth{1.0}
\Line(30,27)(45,-3) \SetWidth{0.5} \Vertex(15,-3){5.66}
\SetWidth{1.0} \Line(15,-3)(30,27) \Line(30,57)(30,87)
\SetWidth{0.5} \Vertex(45,-3){5.66}
\end{picture}}}}

\def\tm54{{\scalebox{0.25}{ 
\begin{picture}(42,72)(8,-38)
\SetWidth{0.5} \SetColor{Black} \Vertex(30,-3){5.66}
\SetWidth{1.0} \Line(30,27)(30,-3) \Line(30,-3)(45,-33)
\SetWidth{0.5} \Vertex(15,-33){5.66} \SetWidth{1.0}
\Line(15,-33)(30,-3) \SetWidth{0.5} \Vertex(45,-33){5.66}
\SetWidth{1.0} \Line(30,-33)(30,-3) \SetWidth{0.5}
\Vertex(30,-33){5.66} \Vertex(30,27){5.66}
\end{picture}}}}

\def\tn55{{\scalebox{0.25}{ 
\begin{picture}(42,72)(8,-38)
\SetWidth{0.5} \SetColor{Black} \Vertex(15,-33){5.66}
\Vertex(45,-33){5.66} \Vertex(30,27){5.66} \SetWidth{1.0}
\Line(45,-33)(45,-3) \SetWidth{0.5} \Vertex(45,-3){5.66}
\Vertex(15,-3){5.66} \SetWidth{1.0} \Line(30,27)(45,-3)
\Line(15,-3)(30,27) \Line(15,-3)(15,-33)
\end{picture}}}}

\def\tp56{{\scalebox{0.25}{ 
\begin{picture}(66,111)(0,0)
\SetWidth{0.5} \SetColor{Black} \Vertex(30,66){5.66}
\Vertex(45,36){5.66} \SetWidth{1.0} \Line(30,66)(45,36)
\Line(15,36)(30,66) \SetWidth{0.5} \Vertex(30,6){5.66}
\Vertex(60,6){5.66} \SetWidth{1.0} \Line(60,6)(45,36)
\SetWidth{0.5}
\SetWidth{1.0} \Line(45,36)(30,6) \SetWidth{0.5}
\Vertex(15,36){5.66}
\end{picture}}}}

\def\tq57{{\scalebox{0.25}{ 
\begin{picture}(81,111)(0,0)
\SetWidth{0.5} \SetColor{Black} \Vertex(45,36){5.66}
\Vertex(30,6){5.66} \Vertex(60,6){5.66} \SetWidth{1.0}
\Line(60,6)(45,36) \SetWidth{0.5}
\SetWidth{1.0} \Line(45,36)(30,6) \SetWidth{0.5}
\Vertex(75,36){5.66} \SetWidth{1.0} \Line(45,36)(60,66)
\Line(60,66)(75,36) \SetWidth{0.5} \Vertex(60,66){5.66}
\end{picture}}}}

\def\tr58{{\scalebox{0.25}{ 
\begin{picture}(81,111)(0,0)
\SetWidth{0.5} \SetColor{Black} \Vertex(60,6){5.66}
\Vertex(75,36){5.66} \SetWidth{1.0} \Line(60,66)(75,36)
\SetWidth{0.5} \Vertex(60,66){5.66}
\SetWidth{1.0} \Line(60,36)(60,66) \Line(60,6)(60,36)
\SetWidth{0.5} \Vertex(60,36){5.66} \Vertex(45,36){5.66}
\SetWidth{1.0} \Line(60,66)(45,36)
\end{picture}}}}

\def\ts59{{\scalebox{0.25}{ 
\begin{picture}(81,111)(0,0)
\SetWidth{0.5} \SetColor{Black}
\Vertex(75,36){5.66} \SetWidth{1.0} \Line(60,66)(75,36)
\SetWidth{0.5} \Vertex(60,66){5.66}
\SetWidth{1.0} \Line(60,36)(60,66) \SetWidth{0.5}
\Vertex(60,36){5.66} \Vertex(45,36){5.66} \SetWidth{1.0}
\Line(60,66)(45,36) \Line(75,6)(75,36) \SetWidth{0.5}
\Vertex(75,6){5.66}
\end{picture}}}}

\def\tt591{{\scalebox{0.25}{ 
\begin{picture}(81,111)(0,0)
\SetWidth{0.5} \SetColor{Black}
\Vertex(75,36){5.66} \SetWidth{1.0} \Line(60,66)(75,36)
\SetWidth{0.5} \Vertex(60,66){5.66}
\SetWidth{1.0} \Line(60,36)(60,66) \SetWidth{0.5}
\Vertex(60,36){5.66} \Vertex(45,36){5.66} \SetWidth{1.0}
\Line(60,66)(45,36) \SetWidth{0.5} \Vertex(45,6){5.66}
\SetWidth{1.0} \Line(45,6)(45,36)
\end{picture}}}}

\def\bigdect{{\scalebox{0.4}{ 
\begin{picture}(140,120)(0,-60)
\SetColor{Black}
\SetWidth{0.5} \Vertex(70,60){5.66}
\put(48,60){\em\huge$\alpha$}
\SetWidth{1.0} \Line(70,60)(0,20)
\SetWidth{0.5} \Vertex(0,20){5.66}
\put(-15,25){\em\huge$\beta$}
\SetWidth{1.0} \Line(70,60)(70,20)
\SetWidth{0.5} \Vertex(70,20){5.66}
\put(50,20){\em\huge$e$}
\SetWidth{1.0} \Line(70,60)(140,20)
\SetWidth{0.5} \Vertex(140,20){5.66}
\put(150,25){\em\huge$\delta$}

\SetWidth{1.0} \Line(0,20)(-50,-20)
\SetWidth{0.5} \Vertex(-50,-20){5.66}
\put(-70,-20){\em\huge $a$}
\SetWidth{1.0} \Line(0,20)(0,-20)
\SetWidth{0.5} \Vertex(0,-20){5.66}
\put(-20,-20){\em\huge$\gamma$}
\SetWidth{1.0} \Line(0,20)(50,-20)
\SetWidth{0.5} \Vertex(50,-20){5.66}
\put(50,-38){\em\huge$d$}

\SetWidth{1.0} \Line(0,-20)(-30,-50)
\SetWidth{0.5} \Vertex(-30,-50){5.66}
\put(-45,-68){\em\huge$b$}
\SetWidth{1.0} \Line(0,-20)(30,-50)
\SetWidth{0.5} \Vertex(30,-50){5.66}
\put(25,-68){\em\huge$c$}

\SetWidth{1.0} \Line(140,20)(100,-10)
\SetWidth{0.5} \Vertex(100,-10){5.66}
\put(80,-20){\em\huge$f$}
\SetWidth{1.0} \Line(140,20)(140,-20)
\SetWidth{0.5} \Vertex(140,-20){5.66}
\put(150,-30){\em\huge$\sigma$}
\SetWidth{1.0} \Line(140,-20)(140,-60)
\SetWidth{0.5} \Vertex(140,-60){5.66}
\put(150,-70){\em\huge$g$}
\SetWidth{1.0} \Line(140,20)(180,-10)
\SetWidth{0.5} \Vertex(180,-10){5.66}
\put(190,-30){\em\huge$\tau$}
\SetWidth{1.0} \Line(180,-10)(180,-60)
\SetWidth{0.5} \Vertex(180,-60){5.66}
\put(190,-70){\em\huge$h$}
\end{picture}}}}

\title{ {Renormalisation and locality:  branched zeta values}}

\author{Pierre Clavier}
\address{Institute of Mathematics,
University of Potsdam,
D-14476 Potsdam, Germany}
\email{clavier@math.uni-potsdam.de}

\author{Li Guo}
\address{Department of Mathematics and Computer Science,
         Rutgers University,
         Newark, NJ 07102, USA}
\email{liguo@rutgers.edu}

\author{Sylvie Paycha}
\address{Institute of Mathematics,
University of Potsdam,
D-14469 Potsdam, Germany\\ On leave from the Universit\'e Clermont-Auvergne\\
Clermont-Ferrand, France}
\email{paycha@math.uni-potsdam.de}

\author{Bin Zhang}
\address{School of Mathematics, Yangtze Center of Mathematics,
Sichuan University, Chengdu, 610064, China}
\email{zhangbin@scu.edu.cn}

\date{\today}

\begin{abstract}
Multivariate renormalisation techniques are implemented in order to build, study  and then renormalise at the poles, branched zeta  functions associated with trees.  For this purpose,  we  first prove   algebraic results and develop  analytic tools,  which we then combine  to study branched zeta functions. The algebraic aspects concern universal properties for  \loc  algebraic structures, some of which had been discussed in previous work;  we  "branch/ lift" to trees operators acting on the decoration set of trees, and factorise branched maps through words  by means of   universal properties for words which we prove in the \loc setup. The analytic tools are  multivariate meromorphic  germs of pseudodifferential symbols
with linear poles which generalise the meromorphic germs of functions  with linear poles  studied in previous work. Multivariate meromorphic germs of pseudodifferential symbols  form a   \loc algebra  on which we build various \loc  maps in the framework of locality structures.  We first show  that the finite part at infinity defines  a \loc character from the latter symbol valued  meromorphic germs to the scalar valued ones. We further equip the \loc algebra of germs of pseudodifferential symbols with \loc Rota-Baxter operators  given by regularised sums and integrals. By means of the universal properties in    the framework of locality structures we can  lift Rota-Baxter operators   to trees, and use the lifted discrete sums in order to   build and study
renormalised branched zeta values  associated with trees. By construction these renormalised branched zeta values factorise on  mutually independent (for the \loc relation) trees.
\end{abstract}

\subjclass[2010]{08A55,16T99,81T15, 32A20, 	52B20}

\keywords{locality, Rota-Baxter algebra, symbols, branched zeta values}

\maketitle

\tableofcontents

\newpage
\section*{Introduction}
Trees offer a useful tool to   understand   the hierarchical structure underlying renormalisation in quantum field theory. They provide  a toy model to
analyse subdivergences arising in Feynman diagrams \cite{K} for  integrals associated with trees reflect the structure of nested
divergences. Our objects of study in this paper are their discrete counterpart,   nested sums associated with trees, which  alongside their
nested structure yield  interesting generalisations of multizeta functions,  which we call branched zeta functions. The latter
generalise  the arborified  zeta values studied  in \cite{M}  using J.~Ecalle's ``arborification" procedure,   viewed as a surjective Hopf algebra morphism
from the Hopf algebra of decorated rooted forests onto a Hopf algebra of shuffles or quasi-shuffle, which amounts to what we call the ``flatening procedure".

Multiple zeta functions $\zeta(s_1,\cdots, s_k)=\sum\limits_{n_1>\cdots >n_k>0} s_1^{-n_1}\cdots s_k^{-n_k}$ can   be interpreted  as sums   associated either
with (Chen) cones or  with rooted (ladder) trees, involving  the pseudodifferential symbols $\sigma_{s_i}, i=1,\cdots$. In \cite{GPZ1}  and
\cite{GPZ2} we generalised multiple zeta functions to sums on general convex cones leading to conical zeta functions.  Here we study their
generalisation to non planar rooted trees leading to branched zeta-functions. In contrast to cones, which offer a relative flexibility  we dealt
with  in~\cite{CGPZ1} using subdivisions, rooted trees present a certain rigidity  and enjoy a universal property which we implement at different stages of the
construction.

Our starting point is  the  Riemann zeta function $\zeta(s)$ defined as the meromorphic extension  {$s\longmapsto\cutoffsum\limits_{n=1}^\infty n^{-s}$} of the  holomorphic map
$s\longmapsto\sum\limits_{n=1}^\infty n^{-s}$ on the half-plane $\Re(s)>1$, {where $\cutoffsum\limits_{n=1}^\infty n^{-s}:=\underset{N\to\infty}{\rm fp}\sum\limits_{n=1}^Nn^{-s}$} is the cut-off regularised sum as defined in
{\cite{P1}}, see also \cite{MP}. Following the same line of thought as \cite{MP}, starting from the (polyhomogeneous)
pseudodifferential symbol $\sigma_s(x)=x^{-s}\,\chi(x)$ on $\R_{\geq 0}$, where $\chi$ is a smooth excision   function at zero, we consider
holomorphic families $z\longmapsto\sigma_{s+z}(x)$ of (polyhomogeneous) symbols on $\R_{\geq 0}$.
The  Riemann zeta function $\zeta(s)$ at a pole $s$,  corresponds to the evaluation at $z=0$
of a regularised cut-off  sum {\begin{equation}\label{eq:sigmas}\zeta(s):= {\rm ev}_{z=0}\left( \underset{N\to\infty}{\rm fp}\sum_{n=1}^N\sigma_{s+z}(n)\right), \quad {\rm
resp.}\quad  \zeta^\star(s)= {\rm ev}_{z=0}\left( \underset{N\to\infty}{\rm fp}\sum_{n=1}^{N-1}\sigma_{s+z}(n)\right).\end{equation}}
{We  reinterpret this expression by means of summation  Rota-Baxter  operators ${\mathfrak S}_\lambda$ with $\lambda\in \{-1,0, 1\}$ (resp.
${\mathfrak I}$), that  to a symbol
$\sigma$ on $\R_{\geq 0}$ assign another symbol
${\mathfrak S}_\lambda(\sigma)$ (resp.
${\mathfrak I}(\sigma)$).
For $\lambda=\pm 1$, the operator ${\mathfrak S}_\lambda$} coincides on any  positive integer $n$  with the discrete summation map
$ n\longmapsto  S(\sigma)(n):=\sum\limits_{k=1}^{n-1}\sigma(k)$ or  $n\longmapsto S(\sigma)(n):=\sum\limits_{k=1}^n\sigma(k)$ according to whether
$\lambda=-1$ or $\lambda=1$. {For $\lambda=0$,  ${\mathfrak S}_0={\mathfrak I}(\sigma)$  is the integral map
$x\longmapsto {\mathfrak I}(\sigma)(x):=\int_0^x\sigma(y)\, dy$ defined for $x\geq 0$.} { For fixed $z\in \C$, we take the finite part at infinity of the map $N\longmapsto  {\mathfrak S}_{\lambda}\left(\sigma_{s+z}\right) (N)$  to build the regularised cut-off sum (compare with (\ref{eq:sigmas}))
\begin{equation} \label{eq:Sigmas}\zeta(s)= {\rm ev}_{z=0}\left(\underset{N\to\infty}{\rm fp}{\mathfrak S}_{-1}\left(\sigma_{s+z}\right) (N)\right), \quad {\rm
resp.}\quad  \zeta^\star(s)={\rm ev}_{z=0}\left(\underset{N\to\infty}{\rm fp}{\mathfrak S}_{1}\left(\sigma_{s+z}\right) (N)\right).\end{equation}   }

This serves as a starting point to build higher zeta functions associated with trees by means of an algebra $\Omega$ (Definition~\ref{defn:Omega}) of (multivariate meromorphic germs  of) symbols on $\R_{\geq 0}$ used to decorate the trees and thereby regularise the multiple integrals and  multiple sums involved in the construction. We adopt a multivariate renormalisation approach already used in  the {toy model of \cite{CGPZ2}}  and  work in the \loc setup discussed in \cite{CGPZ1}. In view of multivariate regularisation,  we consider  the algebra ${\mathcal F}_\Omega$ of	 rooted trees decorated with mutivariate meromorphic  germs of
pseudodifferential symbols  in $\Omega$ and  view  it as an $\Omega$-operated algebra.   The algebra ${\mathcal F}_\Omega$ of rooted trees  is equipped with a
\loc algebra structure involving a partial product on  independent germs,  and an independence relation $\top _{{\mathcal F}_\Omega}$ inherited from the independence relation $\top _\Omega$ on the decorating \loc set $\Omega$.

In the algebraic Part 1 of the present paper, we use a \loc version of   universal properties   {proved in \cite{CGPZ2},} of the  algebra ${\mathcal F}_\Omega$ of	 rooted trees~\cite{F} decorated
by a \loc set $\Omega$. { From \cite{CGPZ2}} we borrow  Corollary \ref{coro:existencemapbranching} which yields a lift of any
\loc map $\phi: \Omega \to \Omega$  on the decoration \loc algebra $\Omega$ to a \loc morphism $\widehat\phi:{\mathcal F}_\Omega\to \Omega $  of
$\Omega$-operated \loc algebras. The first part is mainly dedicated to relating the constructions { on trees of  \cite{CGPZ2}} to new constructions on
\loc algebras of words. It is written in the \loc set category; all the results nevertheless hold in the ordinary set category, which can be recovered
in viewing a set $X$ as a \loc set $(X,\top)$ with the trivial \loc structure $\top=X\times X$. In particular, we
\begin{itemize}
{\item establish a universal property for a \loc Rota-Baxter operator on (proper) words (Theorem \ref{thm:shuffle_free}).}
\item establish a correspondence  (Theorem \ref{thm:stuffle_charac}) between the $\lambda$-Rota-Baxter property of maps  ${\mathfrak S}_\lambda$ on the decoration algebra and
the multiplicative property of theirs lifts to  (decorated) words for a $\lambda$-shuffle product;
\item  introduce a ``flatening" map   in Definition \ref{defn:flatening_map}, which ``flatens" decorated trees to words  corresponding to an ``arborification" \`a la \'Ecalle  discussed in \cite{M};
\item   use the ``flatening" map to relate (Theorem \ref{thm:FW})  the lift to trees of Rota-Baxter operators   with their lift to words.
\end{itemize}

Part 2 is dedicated to the study of the decoration algebra $\Omega$ (Definition \ref{defn:Omega}) of multivariate meromorphic  germs of (polyhomogeneous) pseudodifferential symbols
with linear poles, which contains the algebra ${\mathcal M}$ of  multivariate meromorphic  germs of functions with linear poles introduced in
 \cite{GPZ3}. We equip  $\Omega$  with  an independence relation inherited from an ambient inner product,   similar to the one  defined  on ${\mathcal M}$  in \cite{CGPZ1} and show (Proposition \ref{prop:orderprod})  that it is a \loc algebra for the pointwise product on symbols. We  further prove (Proposition \ref{prop:fpOmega}) that the  renormalised evaluation at infinity
$\underset{+ \infty}{\rm fp}:\Omega\longrightarrow {\mathcal M}$  is a \loc morphism for this independence
relation. Finally, we build ({Theorem \ref{thm:RB_Omega}})   for $\lambda\in \{\pm 1, 0\}$   \loc  $\lambda$-Rota-Baxter operators  on
$\Omega$  which  generalise to the multivariate setup (keeping the same notations), the  summation maps
$ {{\mathfrak S}}_\lambda, \lambda\in \{-1, 1\}$ as well as the integration map ${\mathcal I}$
  on univariate meromorphic germs of symbols mentioned above.

In Part 3, we combine the results of Parts 1 and 2 to build and study branched zeta-functions as multivariate meromorphic functions with linear poles
and their renormalised counterparts. To carry out this programme, we implement  Corollary \ref{coro:existencemapbranching} to build  the
corresponding branched maps which yield \loc morphisms $\widehat{{\mathfrak S}_\lambda} : \mathcal{F}_\Omega \to \Omega$.
  Implementing  the finite part $\underset{+\infty}{\rm fp}$ on $\Omega$  then gives rise to \loc morphisms
  $\underset{+\infty}{\rm fp}\circ \widehat{{\mathfrak S}_\lambda}: \mathcal{F}_\Omega \to  {\mathcal M}$  on the algebra  of trees
  properly  decorated by $\Omega$ to multivariate meromorphic germs of functions. This gives rise to     discrete summation \loc morphisms $\zeta^{{\rm reg},\pm 1}:\mathcal{F}_\Omega \to  {\mathcal M}$ called regularised branched functions~(Proposition \ref{prop:zeta}).  The \loc morphism property ensures the multiplicativity of the regularised branched zeta functions on mutually independent (for the  \loc relation) pairs  of decorated trees (Theorem \ref{thm:Zlocmorph}). In order to study the pole
structure of the regularised branched functions   and investigate the rationality of the renormalised branched zeta values,   we  use the ``flatening"  to express branched zeta-functions
as rational linear combinations of multiple zeta functions.
As a consequence of the linear pole structure of multiple zeta functions,  the poles of any branched zeta function are also linear. Assuming rationality of the
 inner product underlying the multivariate renormalisation procedure, we show that  the  renormalised values at their poles are rational (Theorem
 \ref{thm:rational}).

To conclude, we were able to study the poles of   branched zeta functions thanks to   universal properties of localised Rota-Baxter algebras and  then renormalise the resulting multivariate meromorphic germs by means of a multivariate minimal subtraction scheme we had already implemented in \cite{CGPZ1} for exponential sums on convex cones leading to conical zeta functions. To implement the  multivariate subtraction scheme, we introduced a \loc algebra of multivariate meromorphic germs of polyhomogeneous symbols. The  relative rigidity of tree structures when compared to the
relative flexibility of cone structures, enabled us  to  ``lift" the ordinary discrete summation operator $\sigma\longmapsto \sum_{k=1}^n \sigma(k)$ on polyhomogeneous symbols  to a branched discrete summation operator  on this algebra or meromorphic symbols. This branched summation operator is shown to be a \loc morphism, which ensures multiplicativity of the resulting renormalised branched zeta values on disjoint trees.

So it is the very special tree structure  reflected in the pole structure of the  discrete branched sums of multivariate meromorphic symbols cut-off at infinite (via the finite part map) that enabled us a good control of the poles and hence to renormalise appropriately. The next stage we hope to carry out in forthcoming work is to {provide a precise description of  the tree structure of the  poles } and to  identify a larger class of ``branched multivariate meromorphic germs" that hosts  such cut-off discrete branched sums of multivariate meromorphic symbols to which we can extend similar multivariate minimal subtraction schemes.

\newpage
\part{Algebraic aspects}

An algebraic formulation of the locality principle was provided in~\cite{CGPZ1} in the context of the algebraic approach to perturbative quantum field theory initiated by Connes and Kreimer~\cite{CK}. It was { shown in \cite{CGPZ2}}  that the space spanned by decorated rooted forests equipped with  an appropriate  independence relation inherited from the one on the decorating set,  is the initial object in the category of \loc operated algebras.  We establish  for words similar universal properties,  which we then use to lift  to words
$\lambda$-Rota-Baxter  maps on the decoration algebra.  We further  show that a ``flatening" map -- corresponding to  an ``arborification" procedure \`a la \'Ecalle described in \cite{M}-- which ``flatens" decorated trees to words, defines a \loc map, which we   use to relate their  branched lifts to trees      with their lift to words.

\section{\Loc operated sets and algebras}
\label{sec:loc}

We recall the concepts of \loc structures from \cite{CGPZ1} and \loc operated {structures from \cite{CGPZ2}},  \loc  operated semigroups
and monoids, and \loc operated algebras, successively.

\subsection{\Loc sets, magmas and algebras}

We first recall the concept of a \loc set introduced in \cite{CGPZ1}. {A} {\bf \loc set} is a  couple $(X, \top)$ where $X$ is a set and $ \top\subseteq X\times X$ is a symmetric binary  relation on $X$. For $x_1, x_2\in X$,
denote $x_1\top x_2$ if $(x_1,x_2)\in \top$. We also use the alternative notations $X\times_\top X$ and $X^{_\top 2}$
for $\top$. In general,
for any subset $U\subset X$, let
 			\begin{equation*}
 			U^\top:=\{x\in X\,|\, (x,U)\subseteq \top \}  			 \end{equation*}
denote the {\bf  {polar} subset} of $U$.
For integers $k\geq 2$, {we set}
$$ X^{_\top k}:=X\times_\top \cdots _\top X:= \{(x_1,\cdots,x_k)\in X^k\,|\, x_i \top x_j \text{ for all } 1\leq i\neq j\leq k \}.$$

We call two subsets $A$ and $B$   of a \loc subset $(X,\top )$  {\bf independent}, if
$   A\times B\subset \top.$
  This induces an independence relation on the power set $\mathcal{P}(X)$, which we denote by the same symbol $\top $.
  Then $(\mathcal{P}(X),\top ) $ is a \loc set with $\mathcal{P}(X  )^{\top }=\mathcal{P}(X ^{\top })$.

Recall that two maps  $\Phi,\Psi:\left( X,\top _X\right)\to \left(Y, \top_Y\right)$ are {\bf independent}  and we write $\Phi\top \Psi$ if
$(\Phi\times \Psi)(\top _X) \subseteq \top _Y$, that is,
${x _1\top _X x _2}$ implies $\Phi(x_1)\top _Y\Psi\left(  x_2\right)$ for $x_1,x_2\in X$.
A map  $\Phi:(X,\top _X) \longrightarrow (Y,\top _Y )$   is called a {\bf \loc map} if $\Phi\top\Phi$. Given two \loc sets $(X,\top_X)$ and $(Y,\top_Y)$, let  $\calm or_\top(X,Y)$ denote the set of
			\loc maps from $X$ to $Y$.

We also recall the concepts of \loc monoids and \loc algebras.
The following definition  is  a special instance of a ``partial magma", which is to a magma what a partial algebra is to an algebra~\cite{Gr},   namely a set equipped with a
partial product defined only for certain pairs {with arguments in} the set. See e.g.~\cite{EnM}.
\begin{rk} The condition for a \loc magma is more restrictive than that of a partial magma in that the former requires that the pairs for
which  the partial product is defined stem  from a symmetric relation.
\end{rk}

\begin{defn}\label{defn:magma}
\begin{enumerate}
\item
A  {\bf partial magma}
is a \loc set $(G,\top)$ together with a product law defined on $\top$:
$$ m_G: G\times_\top G\longrightarrow  G
$$
For notational convenience, we usually abbreviate $m_G(x,y)$ by $x\cdot y$ or simply $xy$.
\item
A {\bf sub-partial magma} of a partial magma $(G,\top,m_G)$ is a partial magma $(G',\top',m_{G'})$ with $G'\subseteq G$, $\top'=(G'\times G')\cap \top$ and $m_{G'}=m_G|_{\top'}$, that is, for $x, y\in G'$ and $(x, y)\in \top$, $m_G(x,y)$ is in $G'$.
\item
A partial magma is {\bf commutative} if $m_G(x,y)=m_G(y,x)$ for $(x,y)\in \top$, noting that since $\top$ is symmetric, if one side of the equation is defined, then so is the other.
\item
A {\bf partial semigroup} is {a partial} magma in which the {\bf associativity}
\begin{equation}
(x\cdot y) \cdot z = x\cdot (y\cdot z)
\label{eq:passo}
\end{equation}
holds whenever the expressions on both sides make sense, more precisely if the pairs $(x,y), (x\cdot y,z), (y,z)$ and $(x,y\cdot z)$ are all in $\top$.
\end{enumerate}
 \end{defn}
\begin{ex} \label{ex:semigroups} For a  given subset $A\subset G$ in an arbitrary magma $(G, \star)$,    the relation
$$ \alpha\top_A \beta\Longleftrightarrow \alpha \star \beta\not\in A$$ defines a partial magma.
\end{ex}

We recall the notion of \loc semi-group  introduced in \cite{CGPZ1}.
\begin{defn} \label{defn:lsg}
\begin{enumerate}
\item
A {\bf \loc magma} is a  partial magma $(G,\top,m_G)$  whose product law  is compatible with the \loc relation on $G$ in the following sense:
\begin{equation}\tforall U\subseteq G, \quad  m_G((U^\top\times U^\top)\cap\top)\subset U^\top.
\label{eq:semigrouploc}
\end{equation}
\item
A {\bf \loc semigroup} is a \loc magma whose product law is associative in the following sense:
\begin{equation}
(x\cdot y) \cdot z = x\cdot (y\cdot z) \text{ for all }(x,y,z)\in G\times_\top G\times_\top G. 	
\label{eq:asso}
\end{equation}
Note that, because of the condition \eqref{eq:semigrouploc}, both sides of Eq.~(\mref{eq:asso}) are well-defined for any triple in the given subset.
\label{it:lsg}
\item
A  \loc semigroup is {\bf commutative} if $m_G(x,y)=m_G(y,x)$ for $(x,y)\in \top$.
\item
A  {\bf \loc   monoid} is a \loc   semigroup $(G,\top, m_G)$ together with a {\bf unit element} $1_G\in G$ given by the defining property
\[\{1_G\}^\top=G\quad \text{ and }\quad m_G(x, 1_G)= m_G(1_G,x)=x\quad \tforall  x\in G.\]
We denote the \loc  monoid by $(G,\top,m_G, 1_G)$.
\label{defn:partial monoid}
\end{enumerate}
\end{defn}

\begin{coex}\label{ex:non_assoc_semigroups}
The set $\Q$ equipped with the relation
$$ x\top y\Longleftrightarrow x+y\not\in \Z $$
 is a partial semigroup for the addition $+: \Q\times \Q\to \Q$, but it is {neither} a \loc semi-group {nor a \loc magma}.  Indeed, the \loc condition for semi-groups does not hold: indeed, for $U=\{1/3\}$ we have $(1/3, 1/3)\in ((U^\top\times U^\top)\cap \top)$
 but $ 1/3+1/3=2/3\not \in U^\top$.
 \end{coex}

Here is a related example which will be useful for later purposes.
\begin{coex}\label{ex:non_assoc_semigroups2}
Let  us consider a subset $A $ of $\C$ such that $A+\Z\subset A$. We equip  the power set $\calp(\C)$ of $\C$  with the following relation: for $U, V\in \calp(\C)$,   $U\top_A V\Longleftrightarrow U+V\subset \C\setminus A$. In particular  $U\top_A V\Longrightarrow U+V-\Z_{\ge 0}\subset \C\setminus A$. The locality set $\left(\calp(\C), \top_A \right)$ equipped with the map \begin{eqnarray*} \top_A\subset \calp(\C)\times \calp(\C)&\longrightarrow& \calp(\C)\\
(U, V)&\longmapsto & U+V-\Z_{\ge 0}
\end{eqnarray*}
is a partial semigroup, but not a \loc semi-group.
\end{coex}
\begin{defn}
\begin{enumerate}
\item
An {\bf \loc vector space} is a vector space $V$ equipped with a \loc relation $\top$ which is compatible with the linear structure on $V$ in the
sense that, for any  subset $X$ of $V$, $X^\top$ is a linear subspace of $V$.
\item
Let $V$ and $W$ be vector spaces and let $\top:=V\times_\top W \subseteq V\times W$. A map $f: V\times_\top W \to U$ to a vector space $U$ is called
a {\bf \loc bilinear} map if
$$f(v_1+v_2,w_1)=f(v_1,w_1)+f(v_2,w_1), \quad f(v_1,w_1+w_2)=f(v_1,w_1)+f(v_1,w_2),$$
$$f(kv_1,w_1)=kf(v_1,w_1), \quad
f(v_1,kw_1)=kf(v_1,w_1)$$
for all $v_1,v_2\in V$, $w_1,w_2\in W$ and $k\in \bfk$ such that all the pairs arising in the above expressions are in $V\times_\top W$.
\item A  (not necessarily unitary nor associative) {\bf  \loc  algebra} over $ \bfk$ is a \loc vector space $(A,\top)$ over $ \bfk$ together with a \loc bilinear map
	$$ m_A: A\times_\top A \to A$$ such that
	$(A,\top, m_A)$ is a \loc magma.
\item A  (not necessarily unitary) associative algebra over $ \bfk$ is a \loc vector space $(A,\top)$ over $ \bfk$ together with a \loc bilinear map
	$$ m_A: A\times_\top A \to A$$ such that
	$(A,\top, m_A)$ is a \loc semi-group.
\item A {\bf \loc (unitary and associative) algebra} is a  \loc algebra $(A,\top, m_A)$ together with a {\bf unit} $1_A:  \bfk\to A$ in the sense that
	$(A,\top, m_A, 1_A)$ is a \loc monoid. We shall omit explicitly  mentioning the unit $1_A$ and the product $m_A$ unless this generates an ambiguity.
\end{enumerate}
\end{defn}
Combining  the \loc vector space  and \loc  magma structure, we build \loc algebras and related structures.
\begin{prop}
 Let $(A,\top, m_A)$ (resp. $(A,\top, m_A)$, $(A,\top, m_A, 1_A)$) be a \loc magma (resp. semigroup, algebra). The independence relation $\top$,
 resp. the product $m_A$   on $A$ extends by linearity to  an  independence
 relation  $\widetilde\top$ on $ \bfk\,A$, resp. a bilinear form  $\widetilde m_A:\tilde\top\mapsto  \bfk\,A$.
 Then $( \bfk\,A,\widetilde\top,\tilde m_A)$ (resp. $( \bfk\,A,\widetilde\top,\tilde m_A)$, $( \bfk\,A,\widetilde\top,\tilde m_A,1_A)$) is a nonunitary nonassociative \loc algebra
 (resp. a nonunitary \loc algebra, a \loc algebra).
\end{prop}
\begin{rk} We shall often drop the symbol $\widetilde {}$ over $\top$ denoting the bilinearly extended \loc relation by the same symbol $\top$.
\end{rk}

We next construct free objects in related categories. Let $\calw(\Omega)$ denote the set of words, called {\bf $\Omega$-words}, including the empty word $1$, from the alphabet set $\Omega$. So
\begin{equation}
\calw(\Omega):= \{\omega_1\cdots \omega_k\,|\, \omega_i\in \Omega, 1\leq i\leq k, k\geq 1\}\cup \{1\}.
\mlabel{eq:oset}
\end{equation}
Also let $\calw(\Omega)^*$ denote the set of non-empty $\Omega$-words. So $\calw(\Omega)=\calw(\Omega)^*\cup \{1\}$.

The following results are well-known.

\begin{prop}
Let $\Omega$ be a set. The set $\calw(\Omega)^*$ (resp. $\calw(\Omega)$) is the free semigroup (resp. monoid) on $\Omega$.
\mlabel{pp:freesg}
\end{prop}

We next extend the construction to the \loc context.

\begin{defn}\label{defn:properlydecoratedword}
 Let $(\Omega,\top _\Omega)$ be a \loc set. An $\Omega$-word $w$ is called an {\bf $(\Omega,\top _\Omega)$-proper word}
 if any
 pair of letters in $w$ are independent for $\top_ {\Omega}$. Let   ${\lwords_{\Omega,\top_{\Omega}}}$ denote the set of $(\Omega,\top _\Omega)$-proper words.
We denote the linear span of $\lwords_{\Omega,\top_{\Omega}}$ by $\bfk\,\calw _{\Omega,\top_{\Omega}}$. Similar notions can be defined for the set $\calw(\Omega)^*$ of non-empty $\Omega$-words.
\end{defn}

The set $\lwords_{\Omega}$ is equipped with the independence relation $\top_ {\lwords _\Omega}$ defined for  any pair of words
$ w=\omega_1\cdots\omega_n,w'=\omega'_1\cdots\omega'_{n'}\in \lwords_{\Omega}$ by
\begin{equation*}
 w\top_ {\lwords_ {\Omega}}w' \Longleftrightarrow \{\omega_1,\cdots,\omega_n\} \top _\Omega \{ \omega'_1,\cdots,\omega'_{n'}\}.
\end{equation*}
This {independence} relation restricts to an {independence} relation $\top _{\calw_{\Omega , \top _\Omega}}$ on $\calw _{\Omega , \top _\Omega}$, and extends to an {independence} relation
$\top_ {\bfk\,\calw _ {\Omega}}$ on the linear span $\bfk\,\calw _\Omega $ of $\lwords _\Omega$, the restriction of $\top_ {\bfk\,\calw _ {\Omega}}$ to $\bfk\,\calw _{\Omega,\top_{\Omega}}$ will be denoted by
$\top _{\bfk\,\calw _{\Omega, \top _\Omega}}$. Again similar notions can be defined for the set $\calw(\Omega)^*$ of non-empty $\Omega$-words.

Then Proposition~\ref{pp:freesg} has the following locality variate which can be proved by the same argument.

\begin{prop}
Let $(\Omega,\top)$ be a locality set and let $\calw_{\Omega,\top_\Omega}$ be the locality set of $(\Omega,\top_\Omega)$-proper words. With the concatenation product on $\calw_{\Omega}\times \calw_{\Omega}$ restricted to $\calw_\Omega \times_\top \calw_\Omega$, $\calw_{\Omega,\top_\Omega}$ is a locality monoid. Further it is the free locality monoid on $(\Omega,\top)$, characterised by the universal property: for any locality monoid $(U,\top_U)$ and locality map $f: (\Omega,\top)\to (U,\top_U)$, there is unique morphism $\free{f}: (\calw_{\Omega,\top_\Omega},{\top_{\calw_\Omega}}) \to (U,\top_U)$ such that $\free{f} i = f$ where $i:(\Omega,\top) \to {(\calw_{\Omega,\top_\Omega},\top_{\calw_\Omega})}$ is the natural inclusion.

Similarly, with the restriction of the concatenation product, $\calw^*_{\Omega,\top_\Omega}$ is the free locality semigroup on $(\Omega,\top)$. Furthermore, the linear spans $\bfk\,\calw_{\Omega}$ and $\bfk\,\calw_{\Omega}^*$ are the free locality unitary and nonunitary $\bfk$-algebras on $(\Omega,\top)$.
\mlabel{pp:lfreesg}
\end{prop}

\begin{coro}\label{cor:phisharpword} Let  $(\Omega_i,\top_i), i=1,2$ be two \loc sets. A  \loc map $\phi: \Omega_1\longrightarrow \Omega_2$ uniquely lifts to a \loc monoid morphism
$$
\phi_\calw^\sharp:\calw_{\Omega_1}\longrightarrow \calw_{\Omega_2}$$
and we have
 \[\phi_\calw^\sharp (\omega_1\, w_1){=\phi(\omega_1)\,\phi_\calw^\sharp (w_1)}\quad \text{for all } \omega_1\in \Omega_1, w_1\in \calw_{\Omega_1}.\]
\end{coro}
\begin{proof}
The corollary is a consequence of Proposition~\ref{pp:freelos} applied to $\Omega=\Omega_1, U=\calw_{\Omega_2}$ and $f=\phi i_{\Omega_2}: \Omega_1\to \calw_{\Omega_2}$ where $i_{\Omega_2}: \Omega_2\to \calw_{\Omega_2}$ is the natural inclusion. It follows from the universal property of the $\calw_{\Omega_1}$ that for $\omega_1\cdots\omega_k\in \calw_{\Omega_1}$ with $\omega_1,\cdots,\omega_k\in \Omega_1$, we have
$$ \phi_\calw^\sharp (\omega_1\cdots \omega_k)=\phi(\omega_1)\cdots \phi(\omega_k).$$
This gives the equation in the corollary.
\end{proof}

\subsection{\Loc operated structures}

\begin{defn}
Let $(\Omega, \top)$ be a \loc set. An {\bf $(\Omega,\top)$-operated \loc set} or simply a {\bf \loc operated set} is a \loc set $(X,\top_X)$
together with a partial action on a subset
$\top_{\Omega,X}:=\Omega\times_\top X\subseteq \Omega \times X$

$$\beta:  \Omega\times_\top X \longrightarrow X, \ (\omega, x)\mapsto \beta ^\omega (x)$$
satisfying the following conditions
	   \begin{enumerate}
\item
$\beta\times \Id_X$ yields a map
$$\Omega \times_\top X  \times_\top X
\ola{\beta \times \Id_X} X\times_\top X,$$
where
$$\Omega \times_\top X  \times_\top X: = \{(\omega,u,u')\in \Omega\times X\times X\,|\,
 (u,u')\in \top_X , (\omega, u), (\omega , u')\in \Omega \times_\top X\}.$$

 In other words,
  \begin{equation}
  (\omega,u,u')\in \Omega \times_\top X  \times_\top X \Longrightarrow \beta^\omega(u)\top _Xu'.
\label{eq:graftind0}
\end{equation}
  \item $\Id_\Omega \times \beta$ yields a map
$$ \Omega\times_\top \Omega \times_\top X \ola{\Id_\Omega \times \beta} \Omega \times_\top X.$$
that is, if $(\omega, \omega')\in \top_\Omega, (\omega,u), (\omega',u)\in \Omega\times_\top X$, then
$(\omega',\beta^\omega(u))\in \Omega\times_\top X$.
	\end{enumerate}
 \label{defn:locopset}
 \end{defn}

The sets $\calw(\Omega)^*$ and $\calw(\Omega)$ also have the following universal properties.

\begin{prop}
Let $X$ be a set. The Cartesian product $\calw(\Omega) \times X$ is an $\Omega$-operated set with the operations given by
$$\beta_{\calw(\Omega)\times X}: \Omega \times \calw(\Omega) \times X \to \calw(\Omega)\times X, \quad (\omega, w , x) \longmapsto (\omega\ w,x).$$
Further, together with the map
\begin{equation}
i:X\to  \calw(\Omega)\times X, x\mapsto (1,x),
\mlabel{eq:freeosi}
\end{equation}
$\calw(\Omega) \times X$ is the free $\Omega$-operated set on $X$. More precisely, for any $\Omega$-operated set $(U,\beta_U)$ and set map $f:X\to U$, there is unique homomorphism
$\free{f}:  \calw(\Omega)\times X \to U$ of $\Omega$-operated sets such that $\free{f}\circ i = f$.
\mlabel{pp:freeos}
\end{prop}

Extending Proposition~\ref{pp:freeos} to the locality context needs restriction on the set $X$ because of the compatibility of the locality conditions on the set and on the operator set. But this restriction is broad enough for later applications.

\begin{prop}
Let $(\Omega,\top)$ be a \loc set.
Denote
$$\Omega \times_\top \lwords _{\Omega,\top}^*:=\{(\omega,w)\in \Omega\times \lwords\,|\, (\omega,\omega_1,\cdots,\omega_k)\in \Omega^{_\top (k+1)} \text{ where } w=\omega_1\cdots \omega_k\}.$$
The set $\calw_{\Omega,\top}^*$ is an $(\Omega,\top)$-operated set with the operations given by
$$\beta_{\calw _{\Omega,\top}}: \Omega \times_\top \lwords _{\Omega,\top}^* \longrightarrow \calw _\Omega^*, \quad (\omega, w) \longmapsto \omega\,w.$$
Further, together with the natural inclusion
\begin{equation}
i:\Omega\to  \calw _{\Omega,\top}^*, \omega \mapsto \omega,
\label{eq:freelosi}
\end{equation}
$\calw_{\Omega,\top}^*$ is the free $(\Omega,\top)$-operated set on $\Omega$. More precisely, for any $(\Omega,\top)$-operated set
$(U,\beta_U)$ and any locality map $f:\Omega\to U$, there is unique homomorphism
$\free{f}:  \calw _{\Omega,\top}^* \to U$ of $\Omega$-operated sets such that $\free{f}\circ i = f$.
\mlabel{pp:freelos}
\end{prop}

\begin{proof}
Taking $X=\Omega$ in Proposition~\ref{pp:freeos}, we first note the natural identification
$$ \calw(\Omega) \times \Omega \cong \calw(\Omega)^*, (w,\omega) \mapsto w\omega, (1,\omega)\mapsto \omega \text{ for all } w\in \calw(\Omega), \omega\in \Omega,$$
which allows the identification of $\beta_{\calw(\Omega)\times \Omega}$ in Proposition~\ref{pp:freeos} with
$$ \beta_{\calw(\Omega)^*}:\Omega \times \calw(\Omega)^* \cong  \Omega \times \calw(\Omega) \times \Omega \longrightarrow \calw(\Omega) \times \Omega \cong \calw(\Omega)^*,$$
of which the $(\Omega,\top)$-action $\beta_{\calw_{\Omega,\top}}$ in Proposition~\ref{pp:freelos} is simply the restriction to $\Omega \times_\top \lwords _{\Omega,\top}^*$. Further, the map $i$ in Eq.~(\ref{eq:freeosi}) is identified with the map $i$ in Eq.~(\ref{eq:freelosi}).
With these identifications in mind, let an $(\Omega,\top)$-operated set $(U,\beta_U)$ and a locality map $f:(\Omega,\top)\to (U,\top)$ be given.
Then by Proposition~\ref{pp:freeos}, there is unique morphism $\free{f}:
\calw(\Omega)^*\to U$ such that $\free{f} i =f$. Then restricting $\free{f}$ to $\calw_{\Omega,\top}$ gives the existence of the desired morphism $\free{f}$. Its uniqueness follows from the same inductive proof on the length of the proper $\Omega$-words for Proposition~\ref{pp:freeos}.
\end{proof}

\begin{defn}
 Given a \loc set  $(\Omega , \top )$, we   call a
 \begin{enumerate}
  \item {\bf \loc $(\Omega,\top)$- operated semigroup} a  quadruple  $\left(U,\top _U, \beta ,m_U\right)$ where $(U,\top _U,m_U)$ is a \loc semigroup and
  $\left(U,\top _U, \beta \right)$ is a $(\Omega,\top)$-operated \loc set such that
$$
  (\omega, u, u')\in \Omega\times_\top U\times_\top U \Longrightarrow (\omega, uu')\in \Omega\times_\top U;
  $$
 \item {\bf \loc $(\Omega,\top)$-operated monoid} a  quintuple  $\left(U,\top _U, \beta ,m_U,1_U\right)$ where $(U,\top _U,m_U,1_U)$ is an \loc monoid
  and $\left(U,\top _U, \beta, m_U \right)$ is a $(\Omega,\top)$-operated \loc semigroup, and $\Omega \times 1_U\subset \Omega \times _\top U$.
    \item {\bf $(\Omega,\top)$- operated \loc nonunitary algebra } (resp. {\bf $(\Omega,\top)$- operated \loc unitary algebra}) a  quadruple  $\left(U,\top _U, \beta ,m_U\right)$ (resp. quintuple
    $\left(U,\top _U, \beta ,m_U, 1_U\right)$) which is at the same time a \loc algebra (resp. unitary algebra) and a \loc $(\Omega,\top)$-operated semigroup (resp. monoid), satisfying the additional condition that for any
    $\omega\in \Omega$, the set $\omega^{\top_{\Omega,U}}:=\{ u\in U\,|\, \omega\top_{\Omega,U} u \}$ is a subspace of $U$ on which the action of $\omega$ is linear
(resp. and $\Omega \times 1_U\subset \Omega \times _\top U$). Explicitly,  the linearity condition reads
\begin{quote}
for any $u_1, u_2\in \omega^{\top_{\Omega,U}}$, $k_1, k_2\in \bfk$, we have
$k_1u_1+k_2u_2\in \omega^{\top_{\Omega,U}}$ and  $\beta^\omega(k_1u_1+k_2u_2)=k_1\beta^\omega(u_1)+k_2\beta^\omega(u_2).$
\end{quote}
\end{enumerate}
\end{defn} \label{defn:basedlocsg}

\begin{defn}
A {\bf morphism of \loc operated   \loc structures (sets, semigroups, monoids, nonunitary algebras, algebras)}   $(\Omega_i,\top _i)$-operated structures   $(U_i, \top _{U_i}, \beta_i)$, $i=1,2,$,  is a couple $(\phi, f)$ with $\phi: \Omega_1\to \Omega_2$ and $f: U_1\to U_2$ such that
\begin{itemize}
\item    $\phi$ is a \loc map and $f$ is a morphism of \loc structures;
\item  $f\circ \beta_1^\omega = \beta_2^{\phi(\omega)} \circ f$.
\end{itemize}
\label{defn:morphismoperatedloc}
\end{defn}	

For a given \loc operated set $(\Omega,\top)$, the collection of $(\Omega,\top)$-operated \loc semigroups, with the morphisms in Definition~\mref{defn:morphismoperatedloc} taking $(\Omega_i,\top_i)=(\Omega,\top)$, form a category.

The following lemma will be useful in the sequel.
\begin{lem}\label{lem:Omegaphi}
 Let $(\Omega,\top_\Omega,m_\Omega,1_\Omega)$ be a \loc monoid. A  \loc map $\phi:\Omega\longrightarrow\Omega$ independent of $\Id _\Omega$ induces an $(\Omega,\top^\Omega)$-operated structure on $(\Omega,\top_\Omega,\beta _\phi, m_\Omega,1_\Omega)$,  with
$$		 \beta_\phi:\top _\Omega  \longrightarrow\Omega,
\quad 		 (\omega,\omega') \longmapsto \beta_\phi^\omega( \omega'):=\phi(m_\Omega (\omega,\omega')).
$$
\end{lem}

Thus as a direct consequence of Proposition \ref{pp:freelos}, we obtain

\begin{coro} \label{coro:defn_widehatW}
 Let $(\Omega,\top _\Omega,\bullet)$ be a \loc semigroup. For any locality operator $P:(\Omega,\top_\Omega) \longrightarrow (\Omega,\top_\Omega)$,
there is unique locality map
$$ \widehat{P}^\calw: \calw^*_{\Omega,\top_\Omega} \longrightarrow \Omega$$
such that $\widehat P^{\calw}(\omega):=P(\omega)$ for any $\omega$ in $\Omega$, and $\widehat P^{\calw}(\omega w): = P(\omega\bullet\widehat P^{\calw}(w))$ for any $\omega w\in\calw^*_{\Omega,\top _\Omega}$.
This gives rise to a map
$$  \Phi^{\calw} : \call _\top (\Omega,\Omega)  \longrightarrow \mathcal{L}_{\Omega,\Omega}(\calw^* _{\Omega,\top _{\Omega}},\Omega),
\quad	P \longmapsto \widehat P^{\calw},
$$
 where $\mathcal{L}_{\Omega,\Omega}(\calw^* _{\Omega,\top _{\Omega}},\Omega)$ is the set of morphisms from $\calw ^*_{\Omega,\top _\Omega}$ to $\Omega$ of \loc sets.
 \end{coro}

\subsection{Properly decorated forests and their universal properties}
\label{ss:propdecfor}

	The \loc $\Omega$-operated monoid and \loc algebra of properly decorated rooted forests {are built in \cite{CGPZ2}}. We recall its
	definition and its universal property in the category of \loc operated monoids and algebras.
		 	
\begin{defn} \label{defn:prop_dec_forests}
Let $(\Omega,\top)$ be a \loc set. A {\bf $\Omega$-properly decorated (planar) rooted forest} is a decorated (planar) rooted forest  $\left( F,  d_F \right),$ such that the vertices are decorated by mutually independent elements.

Let  $ \calf _{\Omega,\top _\Omega} $ (resp. $ \calp _{\Omega,\top _\Omega} $) denote the set of  $\Omega$-properly decorated rooted forests (resp. planar rooted forests), and $\bfk\,\calf  _{\Omega,\top _\Omega} $ (resp. $\bfk\,\calp  _{\Omega,\top _\Omega} $) be its linear span. $ \calf _{\Omega,\top _\Omega} $ (resp. $ \calp _{\Omega,\top _\Omega} $)
inherits the independence relation $\top _{\calf_{\Omega}}$ of $\calf_\Omega$ (resp. $\top _{\calp_{\Omega}}$ of $\calp_\Omega$),  which we   denote  by $\top _{\calf_{\Omega, \top _\Omega}}$  (resp. $\top _{\calp_{\Omega, \top _\Omega}}$ ), and $ \bfk\,\calf_{\Omega,\top _\Omega} $
(resp. $\bfk\,\calp_{\Omega,\top_\Omega}$) inherits the independence relation $\top _{\bfk\,\calf_{\Omega}}$ of $\bfk\,\calf_\Omega$ (resp. $\top _{\bfk\,\calp_{\Omega}}$ of $\bfk\,\calp_\Omega$) which we also denote  by $\top _{\bfk\,\calf_{\Omega, \top _\Omega }}$ (resp. $\top _{\bfk\,\calp_{\Omega, \top _\Omega}}$).
\label{defn:properlydecoratedforest}
\end{defn}

It is easy to see that taking disjoint union of forests in $\calf_{\Omega}$ (resp. $\calp_{\Omega}$) defines a \loc monoid structure on
$\calf _{\Omega,\top _\Omega}$ (resp. $\calp _{\Omega,\top _\Omega}$),  and hence a \loc algebra structure in $\bfk\,\calf _{\Omega , \top _\Omega}$
(resp. $\bfk\,\calp _{\Omega , \top _\Omega}$).
This leads
  to the following straightforward yet fundamental result which {we quote from \cite{CGPZ2}.}

\begin{prop} \cite{CGPZ2}
 Let $(\Omega,\top _\Omega)$ be a \loc set.
 Then
\begin{enumerate}
\item
 $(\calf_{\Omega, \top _\Omega},\top _{\calf_{\Omega, \top _\Omega}}, B_+,\cdot,1)$
  is a \loc $(\Omega,\top _\Omega)$-operated commutative monoid;
\item
 $(\bfk\,\calf_{\Omega, \top _\Omega},\top _{\bfk\,\calf_{\Omega, \top _\Omega }}, B_+,\cdot,1)$
is a \loc $(\Omega,\top _\Omega)$-operated  commutative algebra;

 \item
 $(\calp_{\Omega, \top _\Omega},\top _{\calp_{\Omega, \top _\Omega}}, B_+,\cdot,1)$
 is a \loc $(\Omega,\top _\Omega)$- operated  monoid;

 \item
 $\bfk\,\calp_{\Omega, \top _\Omega},\top _{\calp_{\Omega, \top _\Omega}}, B_+,\cdot,1)$
  is a \loc $(\Omega,\top _\Omega)$-operated  algebra.
 \end{enumerate}
\label{prop:operatedalgebra}
\end{prop}

Given   two   \loc  sets   $(\Omega_i, \top _{\Omega_i})$  $i=1,2 $, let
\begin{itemize}
  	\item
  	$\mathcal{L}_{\top }( \Omega_1, \Omega_2)$ denote the set of   \loc maps  $\phi:\Omega_1\longrightarrow \Omega_2$;
  	\item  $\mathcal{L}_{\Omega_1,\Omega_2}(U_1, U_2)$ denote the set of
  	morphisms between $(\Omega_i, \top _{\Omega_i})$- operated \loc structures $(U_i, \top_{U_i}, \beta_i)$ of the same type.
  	\end{itemize}
  		
   All these sets  are equipped with the independence relation of maps: $\phi, \psi: (A,\top _A)\to (B,\top _B)$
  		\[\phi\top \psi\Longleftrightarrow \left(a_1\top _A a_2\Longrightarrow \phi(a_1)\top _B \psi(a_2)\right).\]

\begin{thm}\label{thm:liftedphi}\cite{CGPZ2}
Let $\left(\Omega_1,\top _{\Omega_1}\right),\left(\Omega_2,\top _{\Omega_2}\right)$ be two   \loc sets and let $\phi: \left(\Omega_1,\top _{\Omega_1}\right)\longrightarrow \left(\Omega_2,\top _{\Omega_2}\right)$ be a \loc map.
For any commutative \loc  algebra  $(U, \top _U , \beta _{U} ,m_U, 1_U)$ over $(\Omega_2,\top _{\Omega_2})$, $\phi$ uniquely lifts to
 a morphism  of operated  commutative \loc algebras  $\phi^\sharp: {\mathcal F}_{\Omega_1, \top _{\Omega_1}}\longrightarrow U$, which gives rise to a map
				 	\begin{eqnarray*}
				\sharp: \left(	\mathcal{L} _\top( \Omega_1, \Omega_2),\top\right) &\longrightarrow & \left(\mathcal{L}_{\Omega_1,\Omega_2}(\bfk\,\calf _{\Omega_1,\top _{\Omega_1}}, U),\top\right)\\
					\phi&\longmapsto & \phi^\sharp.
				\end{eqnarray*}

				We call $\phi^\sharp$ the {\bf lifted $\phi$-  map}, which by construction is characterised by the following properties
				  \begin{eqnarray}
					 & \phi^\sharp(\emptyset) = 1_{U}, \label{eq:identity} \\
					 &    \phi^\sharp ((F_1,d_1)\cdots (F_n,d_n)) = \phi^\sharp(F_1,d_1)\cdots\phi^\sharp(F_n,d_n),\label{concatenation}\\
					 &{ \phi^\sharp \left( B_+^\omega(F,d)\right) =  \beta_{U}^{\phi(\omega)}\left(\phi^\sharp(F,d)\right),}  \label{rec_branching_proc}
					 \end{eqnarray}
					 for any mutually independent properly $\Omega_1$-decorated rooted (planar) forests $(F_1,d_1),\cdots, (F_n, d_n) $    and any $\omega\in \Omega_1$ independent of $(F,d)$.
			\end{thm} 

\begin{coro}\label{coro:existencemapbranching}
Let     $\left(\Omega,\top _{\Omega}\right) $ be a commutative \loc  monoid (resp. a unital  commutative \loc  algebra, a \loc  monoid, a unital \loc algebra).
                   A    map  $\phi: \left(\Omega,\top _{\Omega}\right)\longrightarrow \left(\Omega,\top _{\Omega}\right)$ such that \[\phi \top \Id _\Omega, \]
                   induces a unique morphism of \loc commutative monoids (resp. \loc unital  commutative  algebras, \loc  monoids, \loc unital algebras)
   \[   \widehat \phi: \bbf _{\Omega, \top _{\Omega}}\longrightarrow (\Omega, \top_\Omega), \]
   (resp.
  $$\widehat \phi: \bfk\,\calf _{\Omega, \top _{\Omega}}\longrightarrow (\Omega, \top_\Omega), $$
  $$\widehat \phi: \bbp _{\Omega, \top _{\Omega}}\longrightarrow (\Omega, \top_\Omega), $$
  $$\widehat \phi: \bfk \calp _{\Omega, \top _{\Omega}}\longrightarrow (\Omega, \top_\Omega),) $$
$\widehat \phi$ is called  the {\bf branched $\phi$-map}.
 By construction it is characterised by the following properties:      \begin{eqnarray}
                    &\hat  \phi(\emptyset) = 1_{\Omega} \label{eq:identitybranched} \\
                    &  \hat   \phi  ((F_1,d_1)\cdots (F_n,d_n)) = \hat \phi(F_1,d_1)\cdots\hat\phi (F_n,d_n)\label{rec_branching_procbranched}\\
                    &    \widehat \phi \left( B_+^\omega(F,d)\right) =  \phi\left(\omega\,\left(\widehat\phi (F,d)\right)\right), \label{Bbranched}
                    \end{eqnarray}
 for any  mutually independent properly decorated forests where $(F_1,d_1),\cdots, (F_n,d_n)\in \mathcal{F}_{\Omega_1,\top _{\Omega_1}}$, and any $\omega\in \Omega_1$  which is independent of $(F,d)$.
                \end{coro}

\section{From rooted trees to words}

	\subsection{\Loc quasi-shuffle algebra} \label{subsec:loc_struct_words}

\begin{propdefn} \label{propdef:def_star_lambda}
 Let $(\Omega, \top _\Omega, \cdot)$ be a \loc semigroup. For $\lambda\in\bfk$ we define the {\bf $\lambda$-\loc quasi-shuffle product} (for short {\bf  \loc quai-shuffle product} when $\lambda$ is fixed) on $\bfk\,\calw _{\Omega,\top_{\Omega}}$
 $$\star_\lambda:\bfk\,\calw _{\Omega,\top _\Omega}\times_{\top} \bfk\,\calw _{\Omega,\top _\Omega} \to\bfk\calw_{\Omega,\top _\Omega}$$
   as the linear map whose action on the basis elements
 is inductively defined  on words by
 \begin{equation*}
  1\star_{\lambda}w = w\star_{\lambda}1 = w,
 \end{equation*}
 {and for $(\omega,\omega',w,w')\in \lwords_{\Omega,\top _\Omega}^{_{\top}4}$ with $\omega$, $\omega '\in \Omega$
 \begin{equation} \label{eq:def_star_lambda_ind}
  (\omega w)\star_{\lambda}(\omega' w') = \omega(w\star_{\lambda}(\omega' w')) + \omega' ((\omega w)\star_{\lambda}w') + \lambda(\omega\cdot \omega')(w\star_\lambda w')
 \end{equation}}
and extended by bilinearity to $\bfk\,\calw _{\Omega,\top _\Omega}$, this product is well-defined and associative. Thus $(\bfk\,\calw _{\Omega,\top _{\Omega}},\top_ {\bfk\,\calw _{\Omega,\top _{\Omega}}},\star_\lambda,1)$ is an
 \loc algebra.
\end{propdefn}
\begin{proof}
Simply put, $\star_\lambda$ is the restriction of the usual quasi-shuffle product on $\bfk \calw_\Omega$ to independent pairs of words. We give some details to ensure that the restriction is well-defined.

 \begin{claim}
 Let $w$ and $w'$ be independent words. Then
$w\star_\lambda w'$ in Eq.~\eqref{eq:def_star_lambda_ind}
is well defined. For any subset $U\subseteq \bfk \calw_\Omega$,  if $w$ and $w'$ are in $U^\top$, then $w\star_\lambda w'$ is in $U^\top$.
 \label{cl:sh}
 \end{claim}

 We prove the claim by induction on the sum $n$ of the lengths of $w$ and $w'$. If $n=0$ we have $w=w'=1$ and $w\star_\lambda w'=1$ is well defined and is independent of any subset $U$.

For $n=1$, we have $w=\omega \in \Omega$ and $w'=1$ or $w'=\omega ' \in \Omega$ and $w=1$. In each case the product is well defined and independent of any subset that is independent of $w$ {and} $w'$.

Assume that we have shown that the claim has been verified for independent words whose sum of lengths is equal or less than $n\geq1$. Let $w$, $w'$ be two independent words of sum of lengths $n+1$. If $w=1$ or $w'=1$, then
the claim is true by definition. Otherwise  we can write
$w=\omega u$ and $w'=\omega' u'$. 
By the induction hypothesis, the term $ u \star_\lambda w'$ (resp. $ w \star_\lambda u'$, resp.  $u\star_\lambda u'$) is well-defined and is independent of any subset $U$ which are independent of
$u$ and $w'$ (resp. $w$ and $u'$, resp. $u$ and $u'$). Take $V$ to be any subset of $\bfk\calw_\Omega$ independent of $w$ and $w'$.
Then $V\cup\{\omega\}$ is independent of $u$ and $w'$. By the induction hypothesis, $V\cup\{\omega\}$ is independent of $(u\star_\lambda w')$. Then $\omega(u\star_\lambda w')$ is well defined and is independent of $V$. Likewise, the
other two terms on the right hand side of Eq.~\eqref{eq:def_star_lambda_ind}  are well-defined and are independent of $V$. This completes the induction.
\end{proof}

\begin{rk}
 Taking $\lambda=1$ in the above definition gives a locality version of the usual stuffle product. Taking $\lambda=0$ gives the shuffle product.
\end{rk}

\subsection{Free \loc commutative Rota-Baxter algebras}

 We quote further concepts and {results from \cite{CGPZ2}}

\begin{defn}
		A linear operator $P: A\to  A$ on a  commutative \loc algebra $(A,\top)$ over a field $\bfk$ is
		a {\bf \loc Rota-Baxter     operator} of
		{\bf weight} $\lambda\in \bfk$ if it is a \loc morphism, independent of $\Id_A$, and satisfies the following {\bf \loc Rota-Baxter relation}:
		\begin{equation}\label{eq:PRotabaxter}
		P({ a})\, P({ b})= P(P({ a})\,{ b})+ P({ a}\, P({ b})) +\lambda\, P({ a}\,{ b}) \quad \forall (a,b)\in \top.
		\end{equation}
		We call the triple $(A,\top, P)$ a {\bf \loc   Rota-Baxter algebra}.
		
		Let $(A,\top_A,P_A)$ and $(B,\top_B,P_B)$ be two \loc Rota-Baxter algebras of weight $\lambda$. A map $f:A\mapsto B$ is a {\bf Rota-Baxter
		morphism} {if} it is a local algebra morphism such that $f\circ P_A=P_B\circ f$.
	\end{defn}

We now take $\Omega$ to be a locality monoid with unit $1_\Omega$. Note that $1_\Omega$ is not the empty word $1$ of the free monoid $\calw_\Omega$ or the free locality monoid $\calw_{\Omega,\top_\Omega}=(\calw_\Omega, \top_{\calw_\Omega})$.
To avoid confusion, we sometimes use $\sqcup$ to denote the concatenation product in $\calw_\Omega$ in contrast to the product $\cdot$ in $\Omega$. Thus for example, for $w\in \calw_\Omega$, we have $1 \sqcup w=w$ but $1_\Omega \sqcup w \neq w$. Also for $\omega\in \Omega$, we have $1_\Omega \cdot \omega =\omega$, but $1\cdot \omega$ is not defined.

Let $A:=A_\Omega:=\bfk\, \Omega$ be the semigroup algebra. Then $1_A=1_\Omega$.
Denote $\calw_\Omega^*=\calw_\Omega \backslash \{1\}$.

Let us recall two definitions (adapted here to the \loc setting) of \cite[page 93]{G}.
\begin{defn}
Let $(\Omega,\top_\Omega, 1_\Omega)$ be a \loc commutative monoid and let $A:=\bfk\,\Omega$. Let $(A,\top,m_A,1_A)$ be the resulting \loc commutative algebra over $\bfk$. called the {\bf locality monoid algebra} on $(\Omega,\top)$. Denote
$$\sha(A)=\bfk\, \calw_\Omega^*\quad \text{and} \quad \sha_\top(A):=\bfk\,\calw_{\Omega,\top}^*.$$
Define the linear map
 \begin{equation}\label{eq:PA}
 P_A:\sha_\top(A)\to \sha_\top(A), \quad P_A(w) := 1_\Omega w=1_\Omega\sqcup w, w\in \calw_{\Omega,\top}^*.
 \end{equation}

 Let $\diamond_\lambda:\calw_{\Omega,\top}^* \times_{\top_{\calw_\Omega}} \calw_{\Omega,\top}^*\longrightarrow \sha_\top(A)$ be the map defined by
 \begin{equation*}
  (a\sqcup w)\diamond_\lambda(b\sqcup w'):=(a\cdot b)\sqcup (w\star_\lambda w')
 \end{equation*}
 for any proper words $a\sqcup w$, $b\sqcup w'$ in $\calw_{\Omega,\top}$.
By the definition of $\otimes_{\top_{\calw_\Omega}}$ in~\cite{CGPZ1}, $\diamond_\lambda$ extends by linearity to a locality linear map
$$\diamond_\lambda: \sha_\top(A)\otimes_{\top_{\calw_{\Omega}}} \sha_\top(A)\longrightarrow \sha_\top(A).$$
\end{defn}

\begin{rk}
 Since $1_A\in A^\top$ we have $1_A\in(\calw_A^*)^{\top_{\calw_A}}$, therefore $P_A(\calw_{A,\top}^*)\subseteq \calw_{A,\top}^*$ as claimed on the
 definition.
\end{rk}

\begin{thm} \label{thm:shuffle_free}
 Let $(\Omega,\top,m_\Omega,1_\Omega)$ be a locality commutative monoid and let $(A=\bfk\,\Omega,\top,m_A,1_A)$ be the resulting locality commutative monoid algebra over $\bfk$. The quadruple $(\sha_\top(A),\top_{\calw_\Omega},\diamond_\lambda,P_A)$, together with the natural embedding
 $i_A:A \hookrightarrow \sha_\top(A)$ is a free \loc commutative Rota-Baxter algebra of weight $\lambda$ over $A$. More precisely, for any
 \loc commutative Rota-Baxter algebra $(R,\top_R,m_R,P)$ of weight $\lambda$ over $\bfk$ and any \loc algebra homomorphism $f:A\to R$ there is a unique
 homomorphism $\bar f:(\sha_\top(A),\top_{\calw_\Omega},\diamond_\lambda,P_A)\to(R,\top_R,m_R,P)$ of \loc Rota-Baxter algebra of weight $\lambda$  such that
 $f=\bar f\circ i_A$.
\end{thm}

\begin{rk}
 This result is the \loc version, and hence generalisation, of  \cite[Theorem 3.2.1.]{G}. The proof in the locality setup is close to the one of \cite{G}, only making sure at every step that the products are well-defined in view of the partial algebra structures.
\end{rk}

\begin{proof}
In this proof, we shall write $\diamond$ instead of $\diamond_\lambda$ to improve readability.

Throughout this proof, we will repeatedly use that for any pair of nonempty proper words $w$, $w'$ such that $w\top_{\calw_{\Omega}}w'$, each letter in $w$ is independent (in $\Omega$) of each of the letters in
$w'$. In particular, the first letter of $w$ is independent of the first letter of $w'$, and thus $a\cdot b$ is well-defined. Furthermore,
by the semigroup property, $a\cdot b$ is independent of all the words appearing in $w\star_\lambda w'$, showing that $(a\cdot b)\sqcup (w\star_\lambda w')$ is a proper word. We carry out the rest of the proof in several steps.
\smallskip

\noindent
{\bf (I) Well-definedness of $\diamond$:} First, we show that $(\sha_\top(A),\top_{\calw_\Omega},\diamond_\lambda,P_A)$ is a \loc commutative Rota-Baxter algebra of weight $\lambda$ over $\bfk$.
  \begin{itemize}
   \item Since $1_A\in A^\top$ we have $1_A\in(\calw_{\Omega,\top}^*)^{\top_{\calw_{\Omega,\top}}}$. Therefore $P_A\top_{\calw_\Omega}\Id_{\calw_{\Omega,\top}^*}$.
   \item $P_A$ is a Rota-Baxter operator of weight $\lambda$ by direct computation as in the proof of~\cite[Theorem 3.2.1]{G}: let
   $w_a,w_b\in\calw_\Omega^*$, such that $w_a\top_{\calw_\Omega}w_b$. We write $w_a=a\sqcup w_a'$ and $w_b=b\sqcup w_b'$ for some $a,b$ in $\Omega$ and
   $w_a',w_b'$ in $\calw_{\Omega,\top}$. Then we have
   \begin{align*}
    & w_a\diamond  P_A(w_b)  = \left[a\sqcup w_a'\right]\diamond\left[(1_A)\sqcup w_b\right]
    = a\sqcup(w_a'\star_\lambda w_b) \\
    &  P_A(w_a)\diamond w_b = \left[1_A\sqcup w_a\right]\diamond\left[b\sqcup w_b'\right] =   b\sqcup(w_a\star_\lambda w_b'),
   \end{align*}
   Then we have
   \begin{align*}
    P_A(w_a)\diamond P_A(w_b) & = (1_A\sqcup w_a)\diamond_\lambda {(1_A\sqcup w_b) }\quad\text{(by definition of }P_A) \\
			      & = 1_A\sqcup(w_a\star_\lambda w_b) \quad\text{(by definition of }\diamond_\lambda) \\
			      & = 1_A\sqcup\Big(a\sqcup(w_a'\star_\lambda w_b) + b\sqcup(w_a\star_\lambda w_b') \\
			    + & \lambda(a\cdot b)\sqcup(w_a'\star_\lambda w_b')\Big) \quad\text{(by definition of }\star_\lambda) \\
			      & = P_A(w_a\diamond P_A(w_b)) + P_A(P_A(w_a)\diamond w_b) + \lambda P_A(w_a\diamond w_b).
   \end{align*}
  \end{itemize}

\noindent
{\bf (II) Construction of $\free{f}$:} To prove the universal property, let us start by explicitly constructing the map $\bar f$. We let $(B,\top_B, {m_B},P)$  be a unital commutative
  \loc Rota-Baxter algebra of weight $\lambda$ and $f:A\to B$ a
  \loc algebra homomorphism. We inductively define $\bar f:\calw_{\Omega,\top}^*\to B$ from:
  \begin{equation*}
   \bar f(a) := f(a)
  \end{equation*}
  for any proper words $a$ of length $1$ (thus for any $a$ in $\Omega$). Then $\bar f=f\circ i_A$ by definition.

  Now, assume $\bar f$ has been defined on words of length between $1$ and $n\geq1$.  The map $\bar f$ restricts to a \loc map on nonempty proper
  words of length up to $n$ so that for any nonempty proper words $w$, $w'$ of length no larger than $n$, we have $w\top_{\calw_\Omega}w'$ implying $\bar f(w)\top_{\calw_\Omega}\bar f(w')$. For any word $a\sqcup w$ of length $n+1$ we
  define
  \begin{equation*}
   \bar f(a\,\sqcup w) := \bar f(a)\,P(\bar f(w)) = f(a)\,P(\bar f(w)).
  \end{equation*}
  This product in $B$ is well-defined combining the induction hypothesis  with  the {fact} that $f$ and $P$  {and   the composition of local maps are all local maps}. It follows  from the induction assumption combined with the semigroup property, that for any pair of nonempty proper words
  $w$, $w'$ of length no larger than $n+1$,  $w\top_{\calw_\Omega}w'$ implies $\bar f(w)\top_{\calw_\Omega}\bar f(w')$   since both
  $f$ and $P$ are local.

  Thus we have defined a locality map $\bar f:\calw_{\Omega,\top}^*\to B$ and hence $\free{f}: \sha_\top(A)\to B$ by linearity extension.
\smallskip

\noindent
{\bf (III) Compatibility of $\free{f}$ with Rota-Baxter operators:} For any nonempty proper word $w$ a direct computation gives
  \begin{equation*}
   \bar f(P_A(w)) = \bar f(1_A\,\sqcup w) = f(1_A)\,P(\bar f(w)) = 1_B\,P(\bar f(w)) = P(\bar f(w)).
  \end{equation*}
  Thus $\bar f\circ P_A = P\circ \bar f$.
\smallskip

\noindent
{\bf (IV) Multiplicativity of $\free{f}$:}
 We now prove that $\bar f$ is a \loc algebra homomorphism, namely that $\bar f(w\diamond w') =\bar f(w)\diamond \bar f(w')$ by induction on
  $ {n+m\geq 2} $, the sum of the lengths of the nonempty words $w$ and $w'$.
  \begin{itemize}
   \item If $n+m=2$, then $n=m=1$  and we write $w=a$, $w'=b$. Then
   \begin{equation*}
    \bar f(a\diamond b)=\bar f(a\cdot b) =f(a\cdot b)=f(a)\,f(b)=\bar f(a)\,\bar f(b),
   \end{equation*}
   where we have used that $f$ is a \loc algebra homomorphism.
   \item Assume now that for any pair of independent nonempty proper words $w,w'$ whose sum of lengths is no larger than $k$, we have
   $\bar f(w\diamond w') =\bar f(w) \bar f(w')$. Let $w_a$ and $w_b$ be any two independent nonempty proper words whose sum of lengths is
   equal to $k+1$.  Since $P_A$ also acts on empty words,  we can
   write $w_a=(a)\sqcup w$ and $w_b=(b)\sqcup w'$. We use the fact that
   \begin{equation*}
    w_a\diamond w_b = {(a\cdot b)\sqcup[w\star_\lambda w'] = (a\cdot b)\diamond\left[1_A\sqcup[w\star_\lambda w']\right] } = (a\cdot b)\diamond P_A(w)\diamond P_A(w')
   \end{equation*}
   where we have used the associativity of $\diamond$. Then
   \begin{align*}
    \bar f(w_a\diamond w_b) & = \bar f((a\cdot b)\diamond P_A(w)\diamond P_A(w')) \\
			    & = \bar f((a\cdot b)\diamond P_A(w\diamond P_A(w') + P_A(w)\diamond w'+\lambda w\diamond w')).
   \end{align*}
 We now use the fact $P_A$ is a \loc Rota-Baxter operator of weight $\lambda$. Noticing that the image under $P_A$ of a word of length $p$ is a word
   of length $p+1$, and that if $w$ is of length $p$ and $w'$ is of length $q$, then $w\diamond q$ is a word of length $p+q-1$. Thus we  can
  use the induction hypothesis to write 
   \begin{align*}
    \bar f(w_a\diamond w_b) & = f(a\cdot b)\,(\bar f\circ P_A)(w\diamond P_A(w') + P_A(w)\,\diamond w'+\lambda \,w\diamond w') \\
			    & = f(a\cdot b)(\,P\circ\bar f)(w\diamond (1_A\sqcup w') + ((1_A)\sqcup w)\diamond w'+\lambda\, w\diamond w').
   \end{align*}
Using once again   the induction hypothesis we deduce that
   \begin{align*}
    \bar f(w_a\diamond w_b) & = f(a)f(b)\,P\left(\bar f(w)\bar f((1_A)\sqcup w') + \bar f((1_A)\sqcup w)\bar f(w')+\lambda \bar f(w)\bar f(w')\right) \\
			    & = f(a)f(b)P\left(\bar f(w) f(1_A)P(\bar f(w')) + f(1_A)P(\bar f(w))\bar f(w')+\lambda \bar f(w)\bar f(w')\right) \\
			    & = f(a)f(b)P\left(\bar f(w)P(\bar f(w')) + P(\bar f(w))\bar f(w')+\lambda \bar f(w)\bar f(w')\right) \\
			    & = f(a)f(b)P(\bar f(w))P(\bar f(w')) \\
			    & = f(a)P(\bar f(w))f(b)P(\bar f(w')) \\
			    & = \bar f(w_a)\,\bar f(w_b).
   \end{align*}
   This ends our induction.
  \end{itemize}

\noindent
{\bf (V) Uniqueness of $\free{f}$:}
  Finally we prove the uniqueness of $\bar f$ by induction. Assume we have two different such maps $\bar f_1$ and $\bar f_2$. Then for any $a$ in $A$, $\bar f_1(a) = f(a) = \bar f_2(a)$. Thus $\bar f_1$ and $\bar f_2$
  coincide on proper words of length $1$. Assume that $\bar f_1$ and $\bar f_2$
  coincide on proper words of length $n\geq1$. Let $a\sqcup w$ be a nonempty proper word of length $n+1$. Since $a\sqcup w=a\diamond P_A(w)$ we get
  \begin{align*}
   \bar f_1(a\sqcup w) & = \bar f_1(a\diamond P_A(w)) \\
			 & = \bar f_1(a)\,\bar f_1(P_A(w)) \\
			 & = \bar f_2(a)\,P(\bar f_1(w)) \\
			 & = \bar f_2(a)\,P(\bar f_2(w))\quad\text{by hypothesis}\\
			 & = \bar f_2(a)\,\bar f_2(P_A(w)) \\
			 & = \bar f_2(a\sqcup w).
  \end{align*}
  Thus $\bar f_1$ and $\bar f_2$ coincide on $\calw_{\Omega,\top}^*$.
\end{proof}

\subsection{The universal property and quasi-shuffle algebras}

Now let us assume that $(\Omega,\top _\Omega,\bullet)$ is a commutative (non necessarily unital) \loc algebra. Then the map $\widehat P^{\calw}$ of Corollary~\ref{coro:defn_widehatW} can be extended linearly as a linear map
$$\widehat P^{\calw}: \bfk\,\calw ^* _{\Omega,\top _{\Omega}}\to \Omega.
$$

\begin{thm} \label{thm:stuffle_charac}
Let $(G,\top_G,\bullet)$ be a commutative \loc monoid. Let $A:=\bfk G$ and let $(A,\top _A,\bullet, 1_A)$ be the corresponding commutative unital \loc algebra. Let $P:A\longrightarrow A$ be a \loc linear map. The following statements are equivalent:
 \begin{enumerate}
  \item $P$ is a \loc Rota-Baxter operator on $A$ of weight $\lambda$.
  \item $\widehat P^{\mathcal W}:  \left(\bfk\mathcal{W}^*_{A,\top _A},\top,\star_\lambda\right) \longrightarrow  (A,\top _A,\bullet)$
  is a morphism of unital algebras.
Namely, for any mutually independent words $w$ and $w'$ we have
 \begin{equation} \label{eq:stuffle_charac}
  \widehat P^{\mathcal W}(w\star_\lambda w') = \widehat P^{\mathcal W}(w)\bullet\widehat P^{\mathcal W}(w').
 \end{equation}
 \end{enumerate}
\end{thm}

Theorem~\ref{thm:stuffle_charac} can be expressed as the commutative diagram
$$
\xymatrix{
& ({\bfk}\calw_A^*, \diamond_\lambda)\ar^{\overline{\Id}}[dd] & ({\bfk}\calw_A^*, \star_\lambda) \ar_{P_A}[l] \ar_{\widehat{P}^\calw}[ddl] \ar^{\overline{\Id}}[dd] \\
A \ar^{i}[ur] \ar_{\Id}[dr] &&\\
& (A,\bullet)  & (A,\star_\bullet) \ar_{P}[l]
}
$$

\begin{proof} We carry out the proof in the usual setup dropping the \loc, since the proof in the general case is similar, the various \loc
assumptions ensuring that at every step  the   products are well-defined.

 $(\Longleftarrow)$: Assume that Eq.~\eqref{eq:stuffle_charac} holds. Applying it to two words of length 1 gives the Rota-Baxter relation.

 $(\Longrightarrow)$:
 Now if $P_\lambda:A \to A$ is a \loc Rota-Baxter operator of weight $\lambda$, then $(A,\top,P_\lambda)$ is a \loc Rota-Baxter algebra of weight
 $\lambda$. By Theorem \ref{thm:shuffle_free}
 there is a unique homomorphism
 \begin{equation}\label{eq:idbar}\overline{\Id}: ({\bfk} \calw_{A,\top}^*,\diamond_\lambda, P_A) \longrightarrow (A,\bullet, P_\lambda), \end{equation}  { of Rota-Baxter algebras with $P_A$ as in Eq.~(\ref{eq:PA})},  such that $\overline{\Id}\, i =\Id$ for the inclusion
 $i: A\to {\bfk} \calw_{A,\top}^*$ and the algebra morphism $\Id: A\to A$. Thus the composition $\overline{\Id}\,P_A: \bfk \calw^*_{A,\top} \to A$ is a \loc algebra
 homomorphism.

 Let us check that $\overline{\Id}\,P_A = \widehat{P}^{\calw}$.
 First for $\omega\in A$, we have $$\overline{\Id}\,P_A(\omega)=P\overline{\Id}(i(\omega))=P(\omega).$$
 Next for $\omega\,w$ with $\omega\in A, w\in \calw_{A,\top}^*$, we have
 $\omega w=\omega\diamond_\lambda P_A(w)$. Since $\overline{\Id}$ defined in  {Eq.~}(\ref{eq:idbar}) is a homomorphism of \loc Rota-Baxter algebras, we have
 $$ (\overline{\Id}\,P_A)(\omega\,  w)=P(\overline{\Id}(\omega\diamond_\lambda P_A(w)))=P(\overline{\Id}(\omega)\bullet \overline{\Id}(P_A(w))))
 =P(\omega \bullet (\overline{\Id}P_A)(w)).$$
Thus $\overline{\Id}\,P_A$ satisfies the same defining properties as  $\widehat{P}_\lambda^{\calw^*}$, yielding
 $\overline{\Id}\,P_A(\omega)=\widehat{P}_\lambda^{\calw^*}$.

Since \[P_A( w\star_\lambda w') =1_A\, (w\star_\lambda w')= (1_A\, w)\diamond_\lambda (1_A\, w')=  P_A (w)\diamond_\lambda P_A(w'),\] Eq.~\eqref{eq:stuffle_charac} can be verified as follows:
$$ \widehat{P}_\lambda^\calw(w\star_\lambda w') = \overline{\Id} P_A (w\star_\lambda w')=\overline{\Id}\,(P_A(w)\diamond_\lambda P_A(w')) $$ $$=(\overline{\Id}\,P_A(w))\bullet (\overline{\Id}\,P_A(w')) = \widehat{P}_\lambda^\calw(w)\bullet \widehat{P}^\calw(w').$$

Adapting the proof to  the \loc setup,  shows that Eq.~(\ref{eq:stuffle_charac}) holds for any independent $w, w'\in \calw^*_{A,\top}$.
\end{proof}

\subsection{Factorisation through words}
Let $( \bfk\mathcal{W}_{\Omega, \top_\Omega},\top _{\mathcal F}, ,{C_+},\star_\lambda,1)$ be the locality algebra with quasi-shuffle product of weight $\lambda$, introduced in
Proposition-Definition~\ref{propdef:def_star_lambda}. For $\omega\in \Omega$, the map
$$ \beta^\omega_\calw: \bfk\mathcal{W}_{\Omega, \top_\Omega} \to \bfk\mathcal{W}_{\Omega, \top_\Omega}$$
{defined by $w\mapsto \omega w =\omega\sqcup w$ for all $w\in \calw_{\Omega,\top_\Omega}$, and linearly extended to $\bfk\mathcal{W}_{\Omega, \top_\Omega}$} 
defines an $(\Omega,\top_\Omega)$-action on $\bfk\calw_{\Omega,\top_\Omega}$.
Thus applying Theorem~\ref{thm:liftedphi}, we define

 	\begin{defn} \label{defn:flatening_map}
 	The {\bf $\lambda$-flatening operator}
 			\begin{equation}
 			f_\lambda=\Id ^\sharp:(\mathcal{F}_{\Omega, \top_\Omega},\top _{\mathcal F},  B_+,\cdot,1)\longrightarrow ( \bfk\mathcal{W}_{\Omega, \top_\Omega},\top _{\mathcal F}, ,{C_+},\star_\lambda,1)
 			\end{equation}
 			is the unique morphism of $(\Omega,\top_\Omega)$-operated commutative \loc algebras  defined as in Theorem \ref{thm:liftedphi}. In other words, it is characterised by the following properties
 			\begin{eqnarray*}
 				& f_\lambda(1)=1, \label{eq:flatening1}\\
 				& f_\lambda( B_+^\omega(F,d))=\omega\sqcup f_\lambda(F,d), \label{eq:flatening2} \\
 				& f_\lambda((F_1,d_1)\cdot (F_2,d_2)) = f_\lambda(F_1,d_1)\star_\lambda f_\lambda(F_2,d_2). \label{eq:flatening3}
 			\end{eqnarray*}
 		\end{defn}
 		We state a simple, yet important, result concerning the flatening maps.
 		\begin{lem} \label{lem:proper_proper}
 		 Let $(\Omega,\top,.)$ be a \loc semigroup. Then $f_\lambda$ is a \loc map and maps properly decorated forests to linear combinations of properly
 		 decorated words.
 		\end{lem}
 		\begin{proof}
 		 The proof is an easy induction on the number of vertices of the forests. The statement clearly holds for the empty forest. Assuming
 		 it holds for properly decorated forests with $n$ vertices, let $(F,d)$ be a properly decorated forest with $n+1$ vertices.
 		
 		 If $(F,d)=(F_1,d_1)\,(F_2,d_2)$ with $F_i$ nonempty, we have that $f_\lambda(F_1,d_1)\top_\calw f_\lambda(F_2,d_2)$ by the induction  hypothesis   and the result follows since $(\calw_{\Omega,\top},\top_\calw,\star_\lambda)$ is a \loc semigroup.
 		
 		 If $(F,d)=B_+^\omega(F_1,d_1)$ then the result follows from the induction hypothesis by Eq.~\eqref{eq:flatening2} and the
 		 definition of $\top_\calw$.
 		\end{proof}

\begin{rk}{In what follows,   $(G,\top_G,\bullet)$ is  a commutative \loc monoid and we set as before $A:=\bfk G$, which becomes a unital commutative \loc algebra $(A, \top_A, \bullet, 1_A)$. }We extend $\widehat P^{\calw}$ to
 		$\calw _{A,\top_A}$ by setting $\widehat P^{\calw}(1):=1_A$.
 		\end{rk}

 The subsequent theorem states that an operated algebra homomorphism from the free (commutative) operated algebra to a Rota-Baxter algebra
 factors through the free (commutative) Rota-Baxter algebra.
 \begin{thm}\label{thm:FW}
 Given $\lambda \in \bfk$ and  a commutative \loc algebra $(A, \top, \bullet)$, let  $P: A\to A$ be a linear \loc map, $f_\lambda:{\mathcal F}_{A, \top}\to {\mathcal W}_{A, \top}$ be the flatening \loc morphism of  $(A, \top)$-operated  commutative \loc algebras,
  $i_\calw:  (A, \top)\to{\mathcal W}_{A, \top}$ be the natural morphism of  \loc sets and let $\widehat P:{\mathcal F}_A\to A$ be the  \loc morphism of   $(A,\top)$-operated   \loc algebras   built from $P$.  The
  following statements are equivalent
\begin{enumerate}
\item
 The map $\widehat P^{ \calw }:  \left(\mathcal{W}_{A,\top  },\top_{\calw},\star_\lambda\right) \longrightarrow  (A,\top,\bullet) $ is a morphism of commutative $(A , \top  )$-operated \loc algebras;
 \item $P$ is a  \loc $\lambda$-Rota-Baxter operator,
\item {$\widehat P $ factorises through words, that is $\widehat P= \widehat P^{\mathcal W}\circ f_\lambda$.}
\end{enumerate}  In this case, the following diagram
of \loc maps between \loc sets, whose maps in the r.h.s. triangle    are  \loc morphisms of $(A,\top) $-operated  \loc algebras, commutes.
 		
\begin{equation}
\xymatrix{
A \ar@/^3pc/[rrrr]^P \ar@{^{(}->}[rr]^{\iota_\calf} \ar@{^{(}->}[rrdd]^{\iota_\calw} && \calf_{A,\top} \ar[rr]^{\widehat{P}} \ar[dd]^{f_\lambda} && A  \\
&&&&\\
&&\calw_{A,\top} \ar[uurr]^{\widehat{P}^\calw} &&
}
\label{fig:fact_words}
\end{equation}
 		\end{thm}

 		\begin{proof}
 Before proving the equivalence of the assertions, let us briefly comment on the commutativity of  the diagramme. All subdiagrammes outside the r.h.s one commute by construction of the various maps. The commutativity of the r.h.s
 follows from (iii).

 		Assertions (i) and (ii) are equivalent by Theorem \ref{thm:stuffle_charac}. We prove the equivalence of (ii) and (iii).

 $(iii)\Longrightarrow (ii)$: Let us assume that $\widehat P$ factorises through words. Then, for any pair $(\omega,\omega')\in \top _A$ we have on the one hand
 \begin{equation*}
  \widehat P(\bullet_\omega\bullet_{\omega'}) = P(\omega)P(\omega').
 \end{equation*}
 On the other hand we have
 \begin{align*}
  \widehat P(\bullet_\omega\bullet_{\omega'}) & = \widehat P^{\mathcal W}(\omega\star_\lambda\omega') \quad \text{by definition of }f_\lambda \\
					      & = P(\omega P(\omega')) + P(\omega'P(\omega))+\lambda P(\omega\omega')
 \end{align*}
 by definition of $\star_\lambda$ and $\widehat P^{\mathcal W}$, and by linearity of $P$.

$(ii)\Longrightarrow (iii)$: We prove this implication by induction on the number $n$ of vertices of forests. If $n=1$ we directly have $\widehat P(\bullet_\omega) = P(\omega) = \widehat P^{\mathcal W}(\omega)$. Assuming that the result holds
 for all properly decorated forests of at most $n$ vertices, let $(F,d)$ be a properly decorated forest of $n+1$ vertices. If
 $(F,d)= B_+^\omega((F_1, d_1))$ we have
 \begin{align*}
  \widehat P((F,d)) & = P(\omega\bullet\widehat P((F_1, d_1))) \\
			   & = P(\omega\bullet\widehat P^{\mathcal W}(f_\lambda((F_1, d_1)))) \quad\text{ (by the induction hypothesis)} \\
			   & = \widehat P^{\mathcal W}(C_+^\omega(f_\lambda((F_1, d_1)))) \quad\text{ (by definition of } \widehat P^{\mathcal W})\\
			   & = \widehat P^{\mathcal W}(f_\lambda((F,d)))
 \end{align*}
 by definition of $f_\lambda$. If $(F,d)=(F_1,d_1)\,(F_2,d_2)$ with $(F_1, d_1)$ and $(F_2, d_2)$ non-empty we have
 \begin{align*}
  \widehat P((F,d)) & = \widehat P((F_1, d_1))\bullet\widehat P((F_2, d_2)) \\
			   & = \widehat P^{\mathcal W}(f_\lambda((F_1, d_1)))\bullet\widehat P^{\mathcal W}(f_\lambda((F_2, d_2)))  \quad\text{ (by the induction hypothesis)} \\
			   & = \widehat P^{\mathcal W}(f_\lambda((F_1, d_1))\star_\lambda f_\lambda((F_2, d_2))) \quad\text{ (by Theorem \ref{thm:stuffle_charac})} \\
			   & = \widehat P^{\mathcal W}(f_\lambda((F,d)))
 \end{align*}
 by definition of $f_\lambda$. Notice that in both cases, every product is well-defined as we are dealing with \loc maps and by Lemma
 \ref{lem:proper_proper}.
 \end{proof}

 		Let as before  and with the above notations $i_{\mathcal W}$ be the canonical \loc embedding
 		${\mathcal W}_{A, \top}\hookrightarrow {\mathcal F}_{A, \top}$ of words as ladder trees. The following identity follows from the above
 		theorem and the fact that $f_\lambda\circ i_{\mathcal W}=\Id_{\mathcal W}$:
 		\begin{equation}\label{eq:hatPonW} \widehat P\circ i_{\mathcal W}=\widehat P^{\mathcal W}.
 		\end{equation}
 		
Let us recall a result from \cite{CGPZ1}.
 		 {\begin{prop}	{\cite[Proposition 3.22]{CGPZ1}} \label{prop:multpi}
Let $\left(A,\top ,m_A \right)$  be  a \loc algebra. Let $P:A\longrightarrow A$ be a \loc linear idempotent  operator in which case there is a linear decomposition $A=A_1\oplus A_2$ with $A_1={\rm Ker}(\Id-P)$ and
$A_2= {\rm Ker}(P)$ where  $P$ is the projection   onto $A_1$ along $A_2$. The following statements are equivalent:
\begin{enumerate}
\item $P$ is a \loc Rota-Baxter operator; \label{it:irbo1}
\item $A_1$ and $A_2$ are \loc subalgebras of $A$  and $A_1\top A_2$. \label{it:irbo2}
		\end{enumerate} Furthermore, $P$ is a \loc multiplicative map if and only if, in addition to Items~(\ref{it:irbo1}) and (\ref{it:irbo2}), $A_2$ is a \loc ideal of $A$.
\end{prop}}
\begin{coro} Let ${(A,\top _A)}$ be a \loc algebra  and $P: A\to A$ be a {\loc linear} idempotent linear map. The following statements are equivalent.
\begin{enumerate}
\item $A_1:={\rm Ker}(\Id-P)$ and $A_2:= {\rm Ker}(P)$ are \loc subalgebras of $A$  and $A_1\top A_2$.
\item The branched operator  $ \widehat P :  {\mathcal F}_{A,\top _A} \longrightarrow A $   factorises through words.
\end{enumerate}
\end{coro}

 \begin{proof}
   By Proposition \ref{prop:multpi}, the first item is equivalent to $P$ being a Rota-Baxter operator  which  in turn is equivalent to the second item by Theorem \ref{thm:FW}.
    \end{proof}

\begin{ex} \label{ex:mero}
Recall~\cite{CGPZ1} that the  $ {\mathcal M}(\C^\infty)$ of meromorphic germs at zero with linear poles equipped with the
relation $\perp^Q$ induced by an inner product $Q$ is a \loc monoid. In \cite{CGPZ1}, we showed that the inner product $Q$ gives rise to a \loc  projection map $\pi_-^Q:{\mathcal M}(\C^\infty)\longrightarrow {\mathcal M}_-^Q(\C^\infty)$  along the space ${\mathcal M}_+$ of holomorphic germs at zero onto the space ${\mathcal M}_-^Q(\C^\infty)$ of polar germs at zero,  which defines a \loc Rota-Baxter operator.
It follows from the above corollary, that the branched projection map ${\widehat \pi}_-^{Q,{\mathcal F}}: {\mathcal F}_{{\mathcal M}(\C^\infty)}\longrightarrow {\mathcal M}_-^Q(\C^\infty) $ defined on forests decorated by meromorphic germs, factors through  a \loc morphism ${\widehat \pi}_-^{Q,{\mathcal W}}: {\mathcal W}_{{\mathcal M}(\C^\infty)}\longrightarrow {\mathcal M}_-^Q(\C^\infty) $ on words decorated by meromorphic germs.
\end{ex}

\vfill \eject \noindent

\part{Analytic aspects}

In \cite{GPZ3} we studied   $k$-variate meromorphic germs $\sigma(z_1, \cdots, z_k)$ of functions at zero    which form an algebra   generated by compositions $f\circ  \ell$ of  singlevariate meromorphic germs $f(z)$  at zero   composed with non-zero multivariate linear forms $ \ell(z_1, \cdots, z_k)$ on $\C^k$. In the present paper, we consider the algebra  of  symbol-valued $k$-variate meromorphic germs at zero,  generated by compositions $\sigma\circ   \ell$ of symbol-valued singlevariate meromorphic germs $\sigma(z)$  at zero  (the symbols are polyhomogeneous on  some   cone $\Lambda$, here $\R_{\geq 0}$)      composed with non-zero multivariate linear forms $ \ell(z_1, \cdots, z_k)$ on $\C^k$. The latter relate to the former by evaluating the symbol at some point; evaluating a singlevariate meromorphic germ of symbols $\sigma(z)$ at some point $x$ gives rise to a singlevariate meromorphic germ of functions $z\longmapsto \delta_x\circ\sigma(z)$ and a $k$-variate meromorphic germ of symbols $\sigma(z_1, \cdots, z_k)$  gives rise to a $k$-variate meromorphic germ of functions $(z_1, \cdots, z_k)\longmapsto \delta_x\circ\sigma(z_1, \cdots, z_k)$. Letting $x$ tend to $+\infty$, we can build the finite part at infinity, which takes $\sigma(z)$   to a singlevariate meromorphic germ of functions $z\longmapsto \underset{+\infty}{\rm fp}\circ\sigma(z)$, a  $k$-variate meromorphic germ of symbols $\sigma(z_1, \cdots, z_k)$ to a
$k$-variate meromorphic germ of functions $(z_1, \cdots, z_k)\longmapsto \underset{+\infty}{\rm fp}\circ\sigma(z_1, \cdots, z_k)$.   In the context of renormalisation, one can view the composition with non-zero multivariate linear forms as a blow-up to resolve singularities.

The \loc structure $\perp^Q$  on the class of multivariate   meromorphic germs of functions studied in \cite{GPZ3} induces one on the class of multivariate   meromorphic germs of  (polyhomogeneous) symbols defined on    $\Lambda$   via the evaluation map:
\[\sigma\top^Q\sigma':\Longleftrightarrow \left(\delta_x\circ\sigma\right)\, \top^Q\, \left(\delta_x\circ\sigma'\right)\quad \forall x\in \Lambda.\]
We show that the evaluation map $\underset{+\infty}{\rm fp} $   at infinity is a \loc character and that the  integration map and an interpolated summation yield  \loc Rota-Baxter operators on this \loc algebra.

 \section{  \Loc Rota-Baxter operators on   symbols with constrained order}

\subsection{A \loc structure on a class of symbols  with constrained order }

\begin{defn} \label{defn:symbol}
 A smooth function $\sigma:\R_{\geq 0}\to\C$ is called a {\bf symbol} (with constant coefficients) on $\R_{\geq 0}$    if there exists a real number $r$ such that the condition $(C_r)$ below is satisfied.
		\begin{equation}\tag{$C_r$}\label{eq:symbol}
		\forall k\in\Z_{\ge 0},\ \exists D_k \in \R_{> 0}:~\forall x\in\R_{\geq 0}, \ |\partial_x^k\sigma(x)|\leq D_k\langle x\rangle^{r-k}
		\end{equation}
		with $\langle x\rangle:=\sqrt{x^2+1}$.
The set of  symbols on $\R_{\geq 0}$ satisfying the condition $C_r$ is denoted by $\mathcal{S}^r(\R_{\geq 0})$, which is a real vector space.
\end{defn}

Notice that $(r\leq r') \Longrightarrow (C_r \Rightarrow C_{r'})$. So $(r\leq r') \Longrightarrow (\mathcal{S}^r(\R_{\geq 0})\subset \mathcal{S}^{r'}(\R_{\geq 0}))$.

An element of
\begin{equation}\label{eq:Schwartz}
 \cals^{-\infty}(\R_{\geq 0}):= \bigcap_{r\in\R}\cals^r(\R_{\geq 0})
\end{equation}
is called {\bf smoothing}.

\begin{rk} Note that $\cals^{-\infty}(\R_{\geq 0})$ corresponds to the algebra of Schwartz functions on $\R_{\geq 0}$. Thus a symbol is smoothing if and only if it is a Schwartz function.
\end{rk}

By an {\bf excision function} around zero, we mean a smooth function $\chi: \R_{\geq 0}\to \R$ such that  $\chi$ is identically zero in a neighborhood of zero and  identically equal to one outside some interval containing $[0, 1]$\footnote{ {Without loss of generality, we can take an excision function to be identically one outside the unit interval, which we shall do unless otherwise specified}. }. The excision function is there to avoid divergences at zero, the so-called infrared  divergences in physics.
		
\begin{ex} For any complex number $\alpha $, $\langle x\rangle^\alpha $ is a symbol in $\mathcal{S}^{\Re(\alpha)}(\R_{\geq 0})$. For any excision function $\chi$ around zero, and any complex number $\alpha $, $\chi (x)x^\alpha $ is a symbol in $\mathcal{S}^{\Re(\alpha)}(\R_{\geq 0})$. The function $x\longmapsto \log\langle x\rangle$ is an element of $\mathcal{S}^r(\R_{\geq 0})$ for any $r>0$, but it is not an element of $\mathcal{S}^0(\R_{\geq 0})$.
\end{ex}

\begin{prop}
Let $\sigma:\R_{\geq 0}\to\C$ be a symbol. There is at most one pair $(\alpha, \{a_j\})$ with $\alpha\in \C$ and $a_j\in \C, j\in \Z_{\ge 0}, a_0\neq 0,$ such that
\begin{enumerate}
\item $\sigma\in \mathcal{S}^{\Re(\alpha)}(\R_{\geq 0})$, and
\item there is an excision function $\chi$ around zero, such that for any $N\in \Z_{\geq 1}$,  the map
		\begin{equation}\label{eq:sigmaNfixed}x\longmapsto  \sigma_{(N)}^\chi(x):=\sigma(x)-\sum_{j=0}^{N-1}\chi(x)\, a_{ j}\, x^{\alpha-j}
		\end{equation}
	 lies in $\mathcal{S}^{\Re(\alpha)-N}(\R_{\geq 0})$.
	\end{enumerate}
\label{pp:poly}
\end{prop}

\begin {proof} If there are pairs $(\alpha,\{a_j\})$ and $(\beta,\{b_k\})$ with the given conditions, then $\sigma (x)\,\langle x\rangle^{-\alpha}$ converges to the nonzero constant $a_0$ and $\sigma (x)\,\langle x\rangle^{-\beta}$ converges to the nonzero constant $b_0$. This forces $\alpha=\beta$ and $a_0=b_0$.
Then $a_j=b_j, j\geq 1,$ follows inductively  on $j\geq 1$ from the fact  that $$\sigma_{(N)}^\chi(x)\langle x\rangle ^{N-\alpha}\underset{x\to \infty}{\longrightarrow} a_N$$
for any $N\in \Z_{\geq 0}$.

We further notice that
the coefficients $a_{ j}, j\in \Z_{\geq 0}$ are independent of the particular choice of the excision function $\chi$. Indeed, given another excision function $\chi^\prime$, the difference $\sigma_{(N)}^\chi- \sigma_{(N)}^{\chi^\prime}$
 is a Schwartz function.
\end {proof}

\begin{defn}
For a symbol $\sigma:\R_{\geq 0}\to\C$, if the pair in Proposition~\ref{pp:poly} exists, hence is unique, the symbol is  called a {\bf polyhomogeneous }  (also {\bf classical}) symbol of {\bf order} $\alpha$ with {\bf asymptotic expansion}
$\sum_{j=0}^{\infty } a_{ j}\, x^{\alpha-j}$.
We write
\begin{equation}\label{eq:sigmasim}
\sigma(x)\sim\sum_{j=0}^{\infty } a_{ j}\, x^{\alpha-j}.
\end{equation}
The set of  polyhomogeneous symbols on $\R_{\geq 0}$ of order $\alpha $ will be denoted by $S_{\rm ph}^\alpha (\R_{\geq0})$ and its linear span  by $\cals_{\rm ph}^\alpha (\R_{\geq0})$.
\end{defn}

By definition, $S^\alpha_{\rm ph}(\R_{\ge 0})$ and hence $\cals^\alpha_{\rm ph}(\R_{\ge 0})$ are contained in $\cals^{\Re(\alpha)}(\R_{\ge 0})$ and for any $\alpha\in \C$:
$$  k\in \Z_{\geq 0}\Longrightarrow \cals_{\rm ph}^{\alpha-k} (\R_{\geq0})\subset \cals_{\rm ph}^\alpha (\R_{\geq0}).
$$

\begin{ex} For any nonnegative integer $k$ the set
${\mathcal P}^k(\R_{\geq 0})$   of  real polynomial functions of degree $k$  restricted to $\R_{\geq 0}$ is a subset of  $S_{\rm ph}^k (\R_{\geq0})$.
\end{ex}

Since $S_{\rm ph}^\alpha(\R_{\geq 0})\cdot S _{\rm ph}^\beta(\R_{\geq 0})\subset S_{\rm ph}^{\alpha+\beta}(\R_{\geq 0})$, the union $\cup _{\alpha\in\C}S_{\rm ph}^{\alpha}(\R_{\geq 0})$ forms a
	monoid. Let
		\begin{equation*}
		\cals_{\rm ph} (\R_{\geq 0}):=\sum _{\alpha\in\C}\cals_{\rm ph}^{\alpha}(\R_{\geq 0})
		\end{equation*} 	
  be the {linear span of classical symbols} of all orders, then $\cals_{\rm ph} (\R_{\geq 0})$ is an algebra.

\begin {rk}
\begin{enumerate}
\item
The subspace  ${\mathcal S}_{\rm ph}^{ \Z} (\R_{\geq 0})$ generated by polyhomogeneous symbols of {\bf integer order} is a subalgebra of   ${\mathcal S}_{\rm ph} (\R_{\geq 0})$.
\item
The algebra
$\calp (\R_{\geq 0}):= \cup_{k=0}^\infty {\mathcal P}^k(\R_{\geq 0})$   of  real polynomial functions   restricted to $\R_{\geq 0}$ is a subalgebra of $\cals_{\rm ph} (\R_{\geq0})$.	
\item
 The space $ \cals^{-\infty}(\R_{\geq 0})$ forms a subalgebra of $\cals_{\rm ph} (\R_{\geq 0})$ since it lies in $\cals_{\rm ph}^\alpha (\R_{\geq 0})$ for any $\alpha\in \C$. We have $ \cals^{-\infty}(\R_{\geq 0})=\cap_{\alpha\in \C}\cals_{\rm ph}^\alpha (\R_{\geq 0})= \cap_{k\in \Z}\cals_{\rm ph}^k (\R_{\geq 0})$.
\end{enumerate}
\end{rk}

We now define  classes of polyhomogeneous symbols on $\R_{\geq 0}$ with {\bf constrained order}. Given a subset  $A\subset \C$, we consider the linear
 span
 $\cals_{\rm ph} ^{ A}(\R_{\geq 0}):=\sum_{\alpha\in A}\cals_{\rm ph}^{\alpha}(\R_{\geq 0})$ of polyhomogeneous symbols of order in $A$, and we denote by
$\cals_{\rm ph} ^{\notin A}(\R_{\geq 0}):=\sum _{\alpha\in\, {\C}\setminus A}\cals_{\rm ph}^{\alpha}(\R_{\geq 0})$  the linear span of polyhomogeneous symbols of order not in
 $A$.

For a subset  $A\subset \C$ with $A+\Z=A$, we have a direct sum decomposition \begin{equation}\label{eq:directsum}\cals_{\rm ph} ^{\notin A}(\R_{\geq 0}) \oplus \cals_{\rm ph} ^{ A}(\R_{\geq 0}) =\cals_{\rm ph}  (\R_{\geq 0}).\end{equation}
When specializing to $A=\Z$, let
\begin{equation}\label{eq:SigmaRgeq0}\Sigma (\R_{\geq 0}):= \cals_{\rm ph}^{\notin \Z}(\R_{\geq 0})+ {\mathcal P}(\R_{\geq 0});\quad \Sigma(\Z_{\geq 0}):=\{\sigma\vert_{\Z_{\geq 0}}\,|\, \sigma\in \Sigma (\R_{\geq 0})\}.\end{equation}

 \begin{defn}
 For any $A\subset \C$,  consider a relation $\top_A$ on $\cals_{\rm ph}(\R_{\ge 0})$ by
 \begin{equation}\label{eq:orderloc}
  \sigma\top_A\tau\Longleftrightarrow \sigma\cdot \tau\in \cals_{\rm ph} ^{\notin A}(\R_{\geq 0}).
 \end{equation}
More precisely, let $\sigma =\sum \sigma _i$, $\tau=\sum \tau _j$ with $\sigma _i\in S_{\rm ph} ^{\alpha _i}(\R_{\geq 0})$, $\tau _j\in S_{\rm ph} ^{\beta _j}(\R_{\geq 0})$. Then $\sigma\top_A \tau$ means $\alpha _i+\beta _j -\Z_{\ge 0} \cap A=\emptyset$.
\end{defn}

\subsection{The finite part at infinity  }
A symbol in $\cals_{\rm ph} (\R_{\geq 0})$ lies in $L^1(\R_{\geq 0})$ if it is a linear combination of polyhomogenerous symbols of orders with negative real parts,
 in which case we have $ \underset{x\to +\infty}{\lim} \sigma=0 $  as a consequence of (\ref{eq:symbol}).
So,  for  a   symbol  $\sigma $  in $\cals _{\rm ph} (\R_{\geq 0})$  with polyhomogeneous asymptotic expansion given by Eq.~(\ref {eq:sigmasim}), \ { with the notations of  Eq.~\eqref{eq:sigmaNfixed}} we have  $$N>\Re(\alpha) \Longrightarrow \underset{x\to +\infty}{\lim} \sigma_{(N)}^\chi=0. $$
The following definition   taken from  \cite{P1} was also used in \cite{MP}.
\begin{defn}  For  a   symbol  $\sigma $  in $S_{\rm ph}^\alpha (\R_{\geq 0})$ with polyhomogeneous asymptotic expansion {given by Eq.~(\ref {eq:sigmasim})} we set
\begin{equation}\label{eq:finitepart_const}
 \underset{+\infty}{\rm fp} \sigma:=\sum_{j=0}^\infty a_{j}\, \delta_{\alpha- j, 0},
\end{equation}
(with $\delta_{i,0}$ the Kronecker symbol) called the {\bf finite part at infinity} of $\sigma$. Then $\underset{+\infty}{\rm fp}$ can be viewed as a map from $S^\alpha (\R_{\geq 0})$ to $\C$, we extend it by linearity to ${\mathcal S}_{\rm ph} (\R_{\geq 0})$, and we call it the finite part (of infinity) map. We write it $\underset{x\to+\infty}{\rm fp} $ whenever we want to stress the dependence in $x$.
\end{defn}

\begin{rk} The sum  on the r.h.s. is clearly finite since it consists of at most one term, which we refer to as the {\bf constant term}.
\end{rk}

\begin{lem}\label{lem:fpzero} The kernel of the finite part map contains
$\cals_{\rm ph}^{\notin \Z_{\geq 0}} (\R_{\geq 0})$.
\end{lem}
\begin{proof}  The  fact that the finite part at infinity vanishes on
$\cals_{\rm ph}^{{\notin\Z}} (\R_{\geq 0})$ follows  {from Eq.~\eqref{eq:finitepart_const}}
combined with the following trivial observation $\alpha\notin \Z_{\geq 0}\Longrightarrow \alpha-j \notin \Z _{\geq 0} \, \ \forall j\in \Z_{\geq 0}.$
\end{proof}

\begin{ex}
\begin{enumerate}
\item
For   $\sigma \in \cals_{\rm ph}  (\R_{\geq 0})\cap L^1(\R_{\geq 0})=\sum _{\Re(\alpha)<-1}  \cals_{\rm ph}^\alpha (\R_{\geq 0})$,  {we have} $ \underset{+\infty}{\rm fp} \sigma=0$.
\item
For a polynomial $P$ in $ {\mathcal P}(\R_{\geq 0})$,   we have $\underset{  +\infty}{\rm fp}P= P(0)$.
\end{enumerate}
\end{ex}

We investigate the behaviour of the finite part at infinity under pull-back by translations.
For   $a\in \R_+$, let ${\mathfrak t}_a: \R_{\geq 0}\longrightarrow \R_{\geq 0}$ denote the translation $x\longmapsto x+a$.
The   pull-back by the translation map ${\mathfrak t}_a^*: \sigma \longmapsto  \sigma\circ {\mathfrak t}_a$ stabilises ${\mathcal P}(\R_{\geq 0})$  yet
it does not  preserve the finite part  at $+\infty$ on ${\mathcal P}(\R_{\geq 0})$ since for any polynomial $P$, the finite part
$\underset{+\infty}{\rm fp} {\mathfrak t}_a^*P(x)=P(a)  $ depends on $a$. We nevertheless   have the following statement.

\begin{prop}\label{prop:fpinftytransl}For any $a\in \R_+$ and   any $\alpha\in \C$, the  pull-back ${\mathfrak t}_a^*$ by the translation map
stabilises $S_{\rm ph}^\alpha(\R_{\geq 0})$. In general, the finite part $\underset{+\infty}{\rm fp}$ is not invariant under the pull back
$\mathfrak{t}^*_a$, yet its restriction to
$\cals _{\rm ph}^{\notin \Z_{\geq 0}}(\R_{\geq 0})$ is, since $\cals _{\rm ph}^{\notin \Z_{\geq 0}}(\R_{\geq 0})$ lies in the kernel of $ \underset{ +\infty}{\rm fp} \circ {\mathfrak t}_a^*$.
\end{prop}

\begin{proof} The stability of $S_{\rm ph}^\alpha(\R_{\geq 0})$  under the pull-back by the translation  map was shown in  \cite[Proposition 1]{P2} (see also \cite[Proposition 2.52]{P1})  using a Taylor  expansion  at zero in $\frac ax$.  The fact that $\cals _{\rm ph}^{\notin \Z_{\geq 0}}(\R_{\geq 0})$ lies in the kernel of $ \underset{ +\infty}{\rm fp} \circ {\mathfrak t}_a^*$ then follows from the previous lemma.
\end{proof}

The finite part at infinity map defines a character on ${\mathcal P}(\R_{\geq 0})$ since $\underset{+\infty}{\rm fp}(P\, Q)=(P Q)(0)=P(0)\, Q(0)= \underset{+\infty}{\rm fp}(P )\, \underset{+\infty}{\rm fp}(  Q)$. Yet on $ \Sigma (\R_{\geq 0})$, it only defines a partial character.

\begin{prop}\label{prop:fpprod}For  any two symbols $\sigma, \tau\in\Sigma (\R_{\geq 0})$, if $\sigma\top _{\Z}\tau$, then $\underset{+\infty}{\rm fp} \left(\sigma\, \tau\right) = \underset{+\infty}{\rm fp} ( \sigma)   \, \underset{+\infty}{\rm fp}( \tau)=0$.
\end{prop}

\begin{proof} By linearity, we only need to prove this for $\sigma , \tau \in S^\alpha_{\rm ph}(\R _{\ge 0})$ or $\calp^k(\R _{\ge 0})$.
Let  $\alpha$ be the order of $\sigma$,  $\beta$ that of $\tau$ so that  their  product $\sigma\, \tau$ is of order $\alpha+\beta$. By definition, $\sigma\top_\Z \tau $ if and only if  $ \alpha+\beta\notin \Z $, a condition which is fulfilled whenever
  \begin{enumerate}
	\item[(i)] $(\alpha, \beta) \in \left( \Z \times (\C\setminus \Z) \right)\cup \left((\C\setminus \Z )\times \Z   \right)$, or
	\item[(ii)] $\alpha \notin \Z \,\wedge \,\beta \notin \Z   \, \wedge\,  \alpha+\beta\notin \Z  $ holds.
	\end{enumerate}
Both cases are  verified  on the grounds of Lemma \ref{lem:fpzero}.
\end{proof}

\subsection{Differentiation and integration maps on symbols}
 We  single out classes of symbols stable under differentiation and integration, quoting results from \cite{MP} and \cite{P1}.  Clearly, ${\mathcal P}(\R_{\geq 0})$ is stable under differentiation   $\sigma\longmapsto \partial_x\sigma$  and integration $\sigma\longmapsto \int_0^x\sigma$.
\begin{prop}\label{prop:diffint}\begin{enumerate}
\item Differentiation   $\mathfrak{D}: \sigma\longmapsto \partial_x\sigma$ maps $\cals^r(\R_{\geq 0})$ to $ \cals^{r-1}(\R_{\geq 0})$ for any real number $ r$.  It
furthermore maps $\cals_{\rm ph}^\alpha(\R_{\geq 0})$ to $\cals_{\rm ph}^{\alpha-1}(\R_{\geq 0})$ for any $ \alpha$ in $ \C$  and therefore stabilises $\cals _{\rm ph} (\R_{\geq 0})$.
\item  {Integration ${\mathfrak I}: \sigma\longmapsto \int_0^x\sigma$ maps}
\begin{itemize}
\item[(a)]  $\mathcal{S}^r(\R_{\geq 0})$ to $\mathcal{S}^{r+1}(\R_{\geq 0})+\C$ for any real number $r\neq -1$.
\item[(b)]  $\cals_{\rm ph}^{\alpha}(\R_{\geq 0})$ to  $\cals_{\rm ph}^{\alpha+1}(\R_{\geq 0})+\C$ for any ~$\alpha$ in $\C \setminus \Z _{\ge -1}$, so that the  integration map ${\mathfrak I}: \sigma(x)\longmapsto \int_0^x\sigma(t)\, dt$    stabilises $\Sigma (\R_{\geq 0})$.
\end{itemize}

\end{enumerate}

\end{prop}

\begin{proof} \begin{enumerate}
\item  It is easy to check on the grounds of condition (\ref{eq:symbol}) that differentiation $\partial_x$ maps $\mathcal{S}^r(\R_{\geq 0})$  to
$\mathcal{S}^{r-1}(\R_{\geq 0})$ for any real number $r$.

Consequently, for any $\sigma \in \cals_{\rm ph}^\alpha(\R_{\geq 0})$, the remainder term $\sigma_{(N)}^\chi $,   which lies in $\mathcal{S}^{\Re(\alpha)-N}(\R_{\geq 0})$ is mapped
to $\partial_x \sigma_{(N)}^\chi \in \mathcal{S}^{\Re(\alpha)-N-1}(\R_{\geq 0})$. Combining this with
the fact that the homogeneous components $x^{\alpha-j}$ are mapped to  $\partial_xx^{\alpha-j}=(\alpha-j)\,x^{\alpha-j-1}$ of homogeneity degree
$\alpha-1$, and the excision function $\chi$ is mapped to a smooth function $\partial_x \chi$ with compact support, we conclude that
\[\partial_x\sigma= \sum_{j=0}^{N-1} (\alpha-j)\,a_j\,\chi(x)\, x^{\alpha-j-1}+ \sum_{j=0}^{N-1} \partial_x\chi(x)\,a_j\, x^{\alpha-j}+\partial_x \sigma_{(N)}^\chi\sim \sum_{j=0}^{\infty} (\alpha-j)\,a_j\, x^{\alpha-j-1}\] lies  in $\cals_{\rm ph}^{\alpha-1}(\R_{\geq 0})$.

\item (cfr. \cite[Exercise 3.1]{P1}{, see also~\cite[Proposition 2]{MP}})
\begin{itemize}
\item[(a)] For any real number $r< -1$, by condition (\ref{eq:symbol}), we know
$$|\sigma(x)|\leq D_0\langle x\rangle^{r},
$$
so $\int _0^\infty \sigma (y)dy$ converges and
$$\left |\int _0^x \sigma (y)dy-\int _0^\infty \sigma (y)dy\right |=\left|\int _x^\infty \sigma (y)dy\right|\le \int _x^\infty |\sigma (y)|dy\le D_0\int _x^\infty \langle y\rangle^{r}dy.
$$
It is easy to check that $\int _x^\infty \langle y\rangle^{r}dy \in \cals ^{r+1}(\R_{\geq 0})$. Therefore there is a constant $D'_o$ such that
$$\left|\int _0^x \sigma (y)dy-\int _0^\infty \sigma (y)dy\right|\le  D'_0 \langle x\rangle^{r+1}.
$$
Together with the fact that $\partial_x^k \circ \int_0^x= \partial_x^{ k-1}$ for $k\in \Z _{\ge 1}$, we know, $x\longmapsto \int_0^x\sigma(y)\, dy-\int_0^\infty\sigma(y)\, dy$ is in $\mathcal{S}^{r+1}(\R_{\geq 0})$, that is $\int_0^x\sigma(y)\, dy$ is in $\mathcal{S}^{r+1}(\R_{\geq 0})+\C$.

If $r>-1$, then
$$|\sigma(x)|\leq D_0\langle x\rangle^{r}.
$$
So
$$\left|\int _0^x \sigma (y)dy\right|\le \int _0^x |\sigma (y)|dy\le D_0\int _0^x \langle y\rangle^{r}dy.
$$
Now it is easy to check that $\int _0^x \langle y\rangle^{r}dy \in \cals ^{r+1}(\R_{\geq 0})$, thus
$\int_0^x\sigma(y)\, dy$ is in $\mathcal{S}^{r+1}(\R_{\geq 0})=\mathcal{S}^{r+1}(\R_{\geq 0})+\C$.

\item[(b)] Consequently, for any $\sigma \in \cals_{\rm ph}^\alpha(\R_{\geq 0})$ of order $\alpha\notin \Z_{\ge -1}$, and with the notations
of  {Eq.~}(\ref{eq:sigmaNfixed}), the remainder term $\sigma_{(N)}^\chi $ (here $\chi$ is  an excision function  identically one outside the open unit interval), which lies in $\mathcal{S}^{\Re(\alpha)-N}(\R_{\geq 0})$, is mapped to
$\int_0^x \sigma_{(N)}^\chi(y)\,dy \in \mathcal{S}^{\Re(\alpha)-N+1}(\R_{\geq 0})+\C$ when $N\not =\Re (\alpha)+1$, which is the case when $N>\Re (\alpha)+1$.

Now for
$$\sigma(x)=\sum_{j=0}^{N-1}\chi(x)\, a_{ j}\, x^{\alpha-j}+\sigma_{(N)}^\chi(x),
$$
If $\chi(x)$ is identically $1$ when $x\ge x_0$, then for $x\geq x_0$,
\begin{eqnarray*}
 \int_0^x\chi(y)\,y^{\alpha-j}\, dy & =&
 \int_0^{x_0} \chi( y)\,y^{\alpha-j}\, dy+ \int_{x_0}^x\chi(y)\,y^{\alpha-j}\, dy \\
&=&\frac{x^{\alpha-j+1}-x_0^{\alpha -j+1}}{\alpha-j+1}+\int_0^{x_0} \chi( y)\,y^{\alpha-j}\, dy.
\end{eqnarray*}
Thus, for any excision function $\tilde \chi$,   which  vanishes on $[0, r]$ and is identically one  on some interval $[r+\delta, +\infty)$ with $\delta>0$, since  by assumption $\alpha-j\neq -1$ for any $j\in \Z_{\geq 0}$, we have
\[\int_0^x\chi(y)\,y^{\alpha-j}\, dy= \tilde \chi(x)\,\left( \frac{x^{\alpha-j+1}-x_0^{\alpha -j+1}}{\alpha-j+1}+ \int_0^{x_0} \chi( y)\,y^{\alpha-j}\, dy\right)+ (1-\tilde \chi(x))\, \int_0^{x_0}\chi(y)\,y^{\alpha-j}\, dy.\]
 { Thus we have shown that the map $x\longmapsto \int_0^x\chi(y)\,y^{\alpha-j}\, dy$ lies in $\cals_{\rm ph}^{\alpha-j+1}(\R_{\geq 0})+\C$ for any $j\in \Z_{\geq 0}$.

Using Eq.~\eqref{eq:sigmaNfixed} we now write $x\longmapsto \sigma(x)=\sum_{j=0}^{N-1}\chi(x)\, a_{ j}\, x^{\alpha-j}+  \sigma_{(N)}^\chi(x)$.   The above argument in  Part b) tells us that  $\sum_{j=0}^{N-1} \int_0^x \frac{a_j}{\alpha-j+1} \, x_{\alpha-j-1}$
lies  in ${S}_{\rm ph}^{\alpha+1}(\R_{\geq 0})+\C$. Part a) tells us that for large enough $N$,  the symbol $x\longmapsto \int_0^x \sigma_{(N)}^\chi$ lies in   $\mathcal{S}^{\Re(\alpha)-N}(\R_{\geq 0})$. Summing the two we conclude that \[\int_0^x\sigma \sim {\rm constant} +\sum_{j=0}^{\infty} \frac{a_j}{\alpha-j+1} \, x^{\alpha-j-1}\]
  lies  in $\cals_{\rm ph}^{\alpha+1}(\R_{\geq 0})+\C$.    Since the integration  map clearly stabilises ${\mathcal P}(\R_{\geq 0})$, it  follows that it stabilises  $\Sigma (\R_{\geq 0})$.}

Since the integration stabilises ${\mathcal P}(\R_{\geq 0})$, it stabilises  $\Sigma (\R_{\geq 0})$.
\end{itemize}
\end{enumerate}
\end{proof}

\begin{rk}
\begin{itemize}
 \item We want to single out a class of symbols stable under integration: if one insists on avoiding the occurrence of logarithms while integrating more general
symbols, one needs to avoid integrating powers $x^{-1}$, hence the natural  class to consider is  $\Sigma (\R_{\geq 0})$.
\item Let us nevertheless point out   that an alternative point of view adopted in \cite{MP} would be to  extend the algebra of polyhomogeneous symbols to log-polyhomogeneous ones; we chose to avoid this extension which would involve more technicalities.
\end{itemize}
\end{rk}
We saw in Proposition \ref{prop:diffint} that the algebra $\cals_{\rm ph}(\R_{\geq 0})$ is stable under differentiation and  the class $\Sigma (\R_{\geq 0})$  is stable under  integration. So   we can  define the finite part at infinity of an integrated  polyhomogeneous symbol and set the following definition.
\begin{defn}\label{de:cutoffint} For any $\sigma\in \Sigma (\R_{\geq 0})$,\begin{equation}\label{eq:cutoffintclassical}  \cutoffint_0^\infty
\sigma := \cutoffint_0^\infty
\sigma(x)\, d\, x:= \underset{x\to+\infty}{\rm fp}\int_0^x \sigma(y)\, dy\end{equation}
is called the {\bf cut-off integral} of $\sigma$.
\end{defn}
\begin{ex}
We have $ \cutoffint_0^\infty Q =0$ for any polynomial $\sigma=Q  $,  since $\underset{+\infty}{\rm fp}P = P(0)$  vanishes if $P(x)= \int_0^x Q$.
\end{ex}

\begin{ex} \label{ex:cutoffint_vs_int} By the proof of Proposition \ref {prop:diffint}, we know for a classical symbol $\sigma$ of order $<-1$, $\cutoffint_0^\infty\sigma=\int_0^\infty\sigma$.
    \end{ex}

An explicit computation  derived from splitting the integral $\int_0^x=\int_0^{x_0}+\int_{x_0}^x$ for large $x$ and the fact that $\underset{x\to +\infty}{\rm fp}x^{\alpha-j+1}=0$ for $\alpha\neq j-1$ yields the following expression
for any  $\sigma\in \Sigma (\R_{\geq 0})$ of order $\alpha$ (we use the notations of {Eq.~}(\ref{eq:sigmaNfixed})
\begin{equation}\label{eq:Isigma}{\mathfrak I}(\sigma)(x)= \int_0^x \sigma (y)\, dy=\sum_{j=0}^{N-1}a_j\,\left(\int_0^1 \chi(y)\, y^{\alpha-j}
(y)\, dy+ \frac{x^{\alpha-j+1}}{\alpha-j+1} - \frac{1}{\alpha-j+1}\right)+ \int_0^x\sigma_{(N)}^\chi(y)\, dy,
\end{equation}
which after taking the finite part at $+\infty$ yields {(for $N$ sufficiently large)}:
\begin{equation}\label{eq:cutoffint} \cutoffint_0^\infty\sigma (y)\, dy=\sum_{j=0}^{N-1}a_j\,\left(\int_0^1 \chi(x)\, x^{\alpha-j}
(x)\, dx- \frac{1}{\alpha-j+1}\right)+ \int_0^\infty\sigma_{(N)}^\chi(x)\, dx.
\end{equation}
 {This quantity is  clearly independent of the choice of the excision function $\chi$ and  the choice of the integer  $N$ as long as it is chosen sufficiently large.}

\subsection{ Summation of symbols }
To a symbol  $\sigma$  in $\cals_{\rm ph}(\R_{\geq 0})$, we assign to any positive integer $N$ the sum
\begin{equation*}
S(\sigma)(N) := \sum_{n=0}^N\sigma(n),
\end{equation*}
 and for $\lambda\in \{\pm 1\}$ the  sum
	\[   S_\lambda(\sigma)(N)= S(\sigma)\left(N+\frac{\lambda-1}{2}\right), \quad \text{so that}\quad S_{-1}(\sigma)(N)=   S(\sigma)(N-1), \quad S_{1} =   S .\]
		
The Euler-MacLaurin formula (see~\cite[Formula~(13.1.1)]{H}) relates the sum  over $[0,N]\cap \Z$ and the corresponding integral over $[0,N]$. Let $\overline{B_k}(x)= B_k\left(x-[x] \right)$, where $[x]$ stands for the integer part of the real number $x$, and $B_k(x)$ is the $k$-th Bernoulli polynomial. Then
\begin{eqnarray}\label{eq:EML}
S(\sigma)(N)& = &\int_0^N\sigma(x)dx + \frac{1}{2}\left(\sigma(N)+\sigma(0)\right) \nonumber\\
& + & \sum_{k=2}^K\frac{B_k}{k!}\,{\left(\sigma^{(k-1)}(N)- \sigma^{(k-1)}(0)\right)}+ \frac{(-1)^{K+1}}{K!}\int_0^N \overline{B_{K}} (x)\,\sigma^{(K)}(x)\, dx.
\end{eqnarray}
Note that this expression is independent of the choice of the integer $K\geq 2$.\\

Following \cite{MP} we interpolate the discrete sum $S(\sigma)$ by a  smooth function $  {\mathfrak S} (\sigma): \R_{\geq 0}\to \R$ defined as
\begin{eqnarray}\label{eq:EMLS}
{\mathfrak S} (\sigma) & =& {\mathfrak I} (\sigma) + \mu(\sigma),\notag\\
\mu(\sigma)(x)&:=& \frac{1}{2}\left(\sigma(x)+\sigma(0)\right)  + \sum_{k=2}^K\frac{B_k}{k!}\, \,{\left(\sigma^{(k-1)}(x) -  \sigma^{(k-1)}(0) \right)}\label{defi_barP_barQ} \\
&& + \frac{(-1)^{K+1}}{K!}\int_0^x \overline{B_{K}} (t)\,\sigma^{(K)}(t)\, dt\nonumber.
\end{eqnarray}

It follows from  {Eq.~}(\ref{defi_barP_barQ}) that ${\mathfrak S} (\sigma)(N)=S(\sigma)(N)$ for any positive integer $N$.

{\begin{defn}\label{defn:SIlambda} For convenience we define for  $\lambda\in \{\pm 1\}$
\begin{equation}\label{eq:Slambda}
 {\mathfrak S}_\lambda: \sigma\longmapsto \left( x\longmapsto {\mathfrak S} (\sigma)(x+\frac{\lambda-1}{2})\right)
\end{equation}
and set ${\mathfrak S}_0={\mathfrak I}$.
\end{defn}}
 {We have the following generalisation of  \cite[Proposition 8 and Formula(36)]{MP}.}
\begin{prop} \label{prop:SRBloc}
 For  $\lambda\in \{0,\pm 1\}$, the  map ${\mathfrak S}_\lambda$ stabilises $ \Sigma(\Z_{\geq 0}) $.
\end{prop}
\begin{proof}   This conclusion follows from the fact that $\Sigma (\R _{\ge 0})$ is stable under pull-back and ${\mathfrak S}$ stabilises $\Sigma (\R_{\geq 0})$. We now prove that ${\mathfrak S}$ stabilises $\Sigma (\R_{\geq 0})$:
\begin{itemize}
\item { we know from  Proposition \ref{prop:diffint} that} the integration map ${\mathfrak S}_0={\mathfrak I}$  enjoys this property;
\item {the  term $\sum_{k=2}^K\frac{B_k}{k!}\,\sigma^{(k-1) }$ in the Euler-Maclaurin expansion interpreted as a linear combination of  differentiation maps $\partial_x^j$  applied to $\sigma$ also lies in  $\Sigma (\R_{\geq 0})$ if $\sigma$ does}; indeed,  it follows from Proposition \ref{prop:diffint} 2.b. that {$\partial_x^j$ maps} a classical symbol (resp. polynomial) of order $\alpha\in \C\setminus \Z$ (resp. degree $m$) to a classical symbol (resp. polynomial) of order $\alpha-j\in \C\setminus \Z$ (resp. degree $m-j$);
\item Let $\tau_K(x):=\int_0^x \overline{B_{K}} (t)\,\sigma^{(K)}(t)\, dt$. If  $a$ denotes the order of $\sigma$, for any $J\in \Z_{\geq 0}$  we have $\vert \sigma^{(J)}(t)\vert \leq C_J\, \langle t\rangle^{\Re(a)-J}$ for some constant $C_J$. For $K> \Re(a)+1$, the map   $ t\longmapsto \langle t\rangle^{\Re(a)-K-i}$ is $L^1$ for any $i\in \Z_{\geq 0}$.  Since the map $ \overline{B_{K}}$ is $1$-periodic and smooth on any segment $[i, i+1[$, it follows that the map $\tau_K: x\longmapsto \int_0^x \overline{B_{K}} (t)\,\sigma^{(K)}(t)\, dt$ is smooth with
 derivative  $\tau_K^\prime=\overline{B_{K}}  \,\sigma^{(K)}$ which is a classical symbol of order $a-K$ with real part $\Re(a)-K<-1$.  Thus $\tau_K$ differs by a constant $c$ from  a classical symbol of order $a-K+1$. This proves that $\sigma\in \Sigma (\R_{\geq 0})\Longrightarrow\tau_K\in \Sigma (\R_{\geq 0})$ for any  $K> \Re(a)+1$.
\end{itemize}
\end{proof}

\section{\Loc Rota-Baxter operators on symbol-valued multivariate meromorphic germs}

\subsection{The \loc algebra of symbol-valued  meromorphic germs}
We   generalise the space   of meromorphic multivariate  germs of functions  with linear poles at zero considered in \cite{GPZ3}
to the space   of   multivariate   meromorphic  germs    of polyhomogeneous symbols (on $\R_{\geq 0}$)  with linear
poles at zero.

We adopt notations close to those of \cite{GPZ3}. For the filtered Euclidean space \[\left(V, Q\right)=  \underset{\rightarrow}{\lim}\left( \R^k, Q_k\right),\]  let ${\mathcal L}_k:=\left(\R^k\right)^*$ and ${\mathcal L}=\underset{\rightarrow} {\lim}{\mathcal L}_k$ be the direct limit of spaces of linear forms. The space of holomorphic germs at zero with linear poles and real coefficient is denoted by $\calm _+ (\C ^\infty)$, and the space of meromorphic germs at zero with linear poles and real coefficient is denoted by $\calm (\C ^\infty)$.

\begin{defn}\label{defn:multiholsymb}
Let
$U$ be a domain in $\C ^n$.
\begin{itemize}
\item We call a family $(\sigma(z)_{z\in U}) $ of classical symbols
{\bf  holomorphic (of affine order $\alpha (z)$)}  if
\begin{enumerate}\item $\forall z\in U,~\sigma(z)\in \cals _{\rm ph}^{\alpha(z)}(\R_{\geq 0})$;
	\item $ \alpha(z)= L(z)+ c$ with $c\in \R $ and $L\in \call _k$;
\item  for any  $  N\in \Z _{\ge 1}$ and any    excision function $\chi$, the remainder (see Eq.~(\ref{eq:sigmasim}))
		$$~z\mapsto\sigma_{(N)}^{\chi}(z) := \sigma(z)-\sum_{j=0}^{N-1} \chi (x)\, a_{ j}(z)\,x^{\alpha (z)-j},$$
 satisfies the following uniform estimation: for any $k\in \Z_{\geq 0}$, and for any $x\in \R_{\geq 0}$, the derivatives $\partial_x^k\sigma_{(N)}^{\chi}$ are holomorphic functions on $U$, and for any compact subset $K$ of $U$, and any $ n\in \Z_{\geq 0}$   there is a positive constant
		$ C_{k,n,N}(K) $ such that
		\begin{equation}\label{eq:tauNestimate}
		 \left|\partial_z^n\left(\partial_x^k\sigma_{(N)}^{\chi}(z)\right)(x)\right| \leq C_{k,n,N}(K)\langle x\rangle^{\Re (\alpha (z))-N-k+\epsilon}\quad  \forall z\in K\subset U\quad \forall \epsilon>0.
		\end{equation}
 \end{enumerate}

 Such a family is called {\bf a simple holomorphic family of symbols (of affine order $\alpha$}).

	\item Let $\sigma _j (z), j=1, \cdots, J$ be simple holomorphic families of symbols. Then $\sigma(z)=\sum_{j=1}^J \sigma_j(z)$   is called {\bf a holomorphic family of symbols}.	\end{itemize}
\end {defn}

The subsequent straightforward property is nevertheless of importance for the following.

\begin{lem}\label{lem:prodsigma} The product  $\sigma(z):= \sigma_1(z)\, \sigma_2(z)$ of two simple holomorphic families  of symbols $\sigma_i(z) $ of   affine orders $\alpha_i(z)$ is a simple holomorphic family of symbols of affine order $\alpha_1(z)+\alpha_2(z)$.
\end{lem}

\begin{defn} A simple {\bf symbol-valued holomorphic germ} or {\bf holomorphic germ of symbols}  at zero  (with affine order $\alpha(z)$) is an equivalence class of simple holomorphic families  around zero   of symbols of affine order $\alpha(z)$ under the equivalence relation:
$$(\sigma(z)_{z\in U})\sim (\tau(z)_{z\in V})\Longleftrightarrow \exists W,\  0\in W\subset U\cap V, \ \sigma(z)=\tau(z),\ \forall z\in W.
$$
For any positive integer $k$, any $ \alpha(z)= L(z)+ c$ with $ c\in \R,  L\in \call _k$, let  $\calm _+\cals^\alpha(\C^k)$ denote the linear space generated by simple holomorphic germs of symbols of  order $\alpha$, and $\calm _+\cals (\C^k)$ denote the linear space generated by simple holomorphic germs of symbols.
\end{defn}
\begin{rk}
Clearly, two equivalent families of symbols have  the same (affine) order so it makes sense to define the (affine) order of a holomorphic germ of such symbols.
\end{rk}

\begin{defn} Let
$U$ be a domain of $\C ^k$ containing the origin. A simple {\bf meromorphic family} on $U$ of polyhomogeneous symbols with linear poles (with real coefficients) and affine order $\alpha(z)$  is a holomorphic family $(\sigma(z)_{z\in U\setminus X})$  with affine order $\alpha (z)$ of symbols on $U\setminus X$, for which
\begin {itemize}
\item $X=\cup_{i=1}^k \{L_i=0\}$ with $L_1, \cdots L_n\in \call _k$,
\item there exists a simple holomorphic family $(\tau(z)_{z\in U})$ with affine order  $\alpha(z)$ and nonnegative integers $s_1, \cdots, s_n$, such that
$$L_1^{s_1}\cdots L_n^{s_n}\sigma (z)=\tau(z)
$$
on $U\setminus X$.
\end{itemize}
A simple symbol-valued {\bf meromorphic germ} or meromorphic germ of symbols at zero on $\C^k$ with linear poles and   affine order $\alpha(z)$    is an equivalence class of meromorphic families  around zero  with linear poles of  symbols of affine order $\alpha(z)$    under the equivalence relation:
$$(\sigma(z)_{z\in U\setminus X})\sim (\tau(z)_{z\in V\setminus Y})\Leftrightarrow \exists W,\  0\in W\subset U\cap V, \ \sigma(z)=\tau(z),\ \forall z\in W\setminus (X\cup Y),$$
where
$U$  and $V$ are domains of $\C^k$ containing the origin.
Let ${\mathcal M}{\mathcal S}^\alpha(\C^k)$ denote the linear space generated by simple symbol-valued meromorphic germ at zero on $\C^k$ with linear poles and affine order $\alpha(z)$, and $\calm \cals (\C^k)$ denote the linear space generated by simple symbol-valued meromorphic germ at zero.
\end{defn}

Composing with the projection $(\C^{k+1})^* \to (\C^k)^*$   dual to the inclusion $\iota_k:\C^k\to \C^{k+1}$, and the isomorphism induced by the inner product $Q^*_k: (\C^{k})^* \cong \C^{k}$,  yields the embeddings
${\mathcal M}_+ {\mathcal S}(\C^k)\hookrightarrow {\mathcal M}_+ {\mathcal S}(\C^{k+1})$ (resp. ${\mathcal M} {\mathcal S}(\C^k)\hookrightarrow {\mathcal M} {\mathcal S}(\C^{k+1})$), thus giving rise to the direct  limits:

\begin{eqnarray} {\mathcal M}_+^\alpha {\mathcal S}(\C^\infty) &:= &\underset{\rightarrow} {\lim}{\mathcal M}_+^\alpha {\mathcal S}(\C^k) =\bigcup_{k=1}^\infty {\mathcal M}_+ ^\alpha{\mathcal S}(\C ^k),\\
{\mathcal M}_+  {\mathcal S}(\C^\infty) &:= &\underset{\rightarrow} {\lim}{\mathcal M}_+ {\mathcal S}(\C^k) =\bigcup_{k=1}^\infty {\mathcal M}_+ {\mathcal S}(\C^k),\\
	\left({\rm resp. }\ {\mathcal M}  {\mathcal S}(\C^\infty)\right.&:= & \underset{\rightarrow}{\lim}{\mathcal M}  {\mathcal S}(\C^k)=\bigcup_{k=1}^\infty {\mathcal M} {\mathcal S}(\C^k) ,
\end{eqnarray}
where  $ \alpha(z)= L(z)+ c$ with $c\in \R $ and $L\in \call _k$.

\begin{ex} \label{ex:holo_family_BZV}
	For any $\ell\in {\mathcal L}_k$,
	$$z\longmapsto \left(x\mapsto\langle x\rangle ^{\ell(z)}\right) $$
	defines a simple holomorphic germ of symbols of  order $\ell$.
\end{ex}

\begin {prop} Under  pointwise function multiplication,  $ {\mathcal M}{\mathcal S}(\C^\infty)$ is a complex algebra  and  we have the following inclusions of  subalgebras   \[ {\mathcal M}_+\cals (\C ^\infty)\subset {\mathcal M}{\mathcal S}(\C ^\infty); \quad  {\mathcal M}   (\C ^\infty)\cdot {\mathcal P}(\R_{\geq 0})\subset{\mathcal M} {\mathcal S}  (\C ^\infty).\]
\end{prop}

\subsection {Dependence space of a meromorphic germ of symbols }
The notion of dependence space defined in \cite[Definition 2.13]{CGPZ1} for meromorphic germs extends   to meromorphic germs of symbols since the arguments used there to justify the definition   apply in the same way.
\begin {defn}
Let $\sigma (z)$ be a meromorphic family of symbols with affine order on an open neighborhood $U$ of $0$ in $\C ^n$. If there are linear forms $L_1, \cdots , L_k$ on $\C ^n$ and a meromorphic family of symbols $\tau (w) $ on an open neighborhood $W$ of $0$ in $\C ^k$, such that $\sigma(z) =\tau(\phi(z))$ on $U\cap \phi ^{-1}(W)$, where $\phi=(L_1, \cdots, L_k): \C ^n \to \C ^k$, then we say that {\bf $\sigma $ depends} on the (linear) subspace of $(\C ^n)^*$ spanned by $L_1, \cdots, L_k$. We say that a meromorphic germ of symbols at $0\in \C ^n$ depends on a subspace $W \subset (\C ^n)^*$ if one of its representatives in the equivalence class given by the germ  does.

 The {\bf dependence subspace} $\supp(\sigma)$ of a meromorphic family of symbols $\sigma$ on an open neighborhood $U$ of $0$ in $\C^n$,  is the smallest subspace of $(\C^n)^*$   {on which it depends}. For a meromorphic germ, the dependence subspace is the dependence subspace of any of its representing element.
\end{defn}

\begin{lem}\label{lem:stabdiffint} For any $\sigma\in \Omega$ we have \[\supp(\partial_x^\alpha\sigma)\subset \supp( \sigma); \quad
\supp\left({\mathcal I}(\sigma)\right)\subset \supp( \sigma), \]
where ${\mathcal I}$ is the integration map defined in Eq.~(\ref{eq:Isigma}).
\end{lem}
\begin{proof} The first inclusion follows from the fact that differentiation commutes with multiplication by meromorphic germs of functions. The second inclusion follows from inspection of the explicit formula (\ref{eq:Isigma}) for ${\mathfrak I}$.
\end{proof}

\begin {prop}
\label{prop: LocCoeff}
If a meromorphic germ of symbols $\sigma (z)$ depends on  a space $V$, and
$$\sigma (z) \sim \sum a_{n}(z)x^{\alpha (z)-n}$$
then so do $\alpha (z)$, $a_n(z), n\in \Z _{\ge 0}$ depend on $V$.
\end{prop}
\begin{proof} Assuming that the polyhomogeneous symbol $\sigma (z) \sim \sum a_{n}(z)\, x^{\alpha (z)-n}$ reads $\sigma(z)=\tau(\phi(z)) $ for some symbol $\tau(z)\sim\sum b_n (z)\, x^{\beta(z)-n}$, it follows  from the uniqueness of the coefficients in the asympotic expansion of a polyhomogeneous symbol that
$a_n=b_n\circ \phi$ and $\alpha= \beta\circ \phi$ for any non negative integer $n$. The statement then follows.
\end{proof}

We are now ready to extend the  independence relation on meromorphic germs of functions  introduced in \cite[Definition 2.14]{CGPZ1} to meromorphic germs of symbols.
\begin{defn}
Two  meromorphic germs of symbols $\sigma_1 $ and $\sigma_2$ in $\Omega$ are said to be {\bf independent} whenever \[\sigma_1 \perp ^Q \sigma_2 :\Longleftrightarrow  \supp(\sigma_1 )\perp ^Q \supp(\sigma_2).\]
\mlabel{de:meroindep}
\end{defn}

\begin{ex}\begin{enumerate}
\item Given two  functions  $f, g\in {\mathcal M}(\C^\infty)$ and two polynomials $P, Q$ on $\R_{\geq 0}$, we have
\[ (f\cdot P)\perp^Q (g\cdot Q)\Longleftrightarrow f\perp^Q g.\]
\item Given $\ell_i, L_i \in {\mathcal L}$ with $i=1,2$, we have \[\frac{\langle x\rangle^{\ell_1}}{L_1}\perp^Q \frac{\langle x\rangle^{\ell_2}}{L_2}\Longleftrightarrow \{\ell_1, L_1\} \perp^Q \{\ell_2, L_2\}.\]
\end{enumerate}
\end{ex}

\begin{defn} We call the complex linear space generated by the set ${\cup_{\alpha^\prime (0)\neq 0}{\mathcal M}{\mathcal S}^\alpha(\C^k)}$ and the linear space  ${\mathcal M} (\C^\infty)\cdot {\mathcal P}(\R_{\geq 0})$ the space of admissible meromorphic germs of symbols, and denote it by $\Omega$.
\label{defn:Omega}
\end{defn}

\begin{rk} \begin{enumerate} \label{rk:cancelling orders}\item \label{rk:cancelling orders i}$\Omega$  is not an algebra for $c+c'\notin \Z_{\geq 0}\Longrightarrow \chi(x) \, x^{qz+c}\,   x^{-qz+c'}\notin \Omega$ for any excision function $\chi$ around zero. Yet it can be equipped with a locality algebra structure as we shall see from the subsequent proposition.
\item  The fact that we exclude meromorphic germs of symbols with constant order which are not polynomial, is motivated by the  fact that allowing for negative integer powers  can give rise to logarithmic symbols after integration. Here like in quantum field theory, we want to avoid such logarithmic symbols.
\end{enumerate}
\end{rk}
\begin{prop}\label{prop:orderprod} The triple
		$\left(\Omega, \perp^Q, m_{\Omega}\right)$ is a  commutative and unital \loc algebra, with unit given by the constant function $1$ and
$m_{\Omega} $ is the restriction of	the pointwise function multiplication to  the graph $\perp^Q\subset \Omega\times \Omega$.
	\end{prop}

\begin{proof} Assume that $\sigma_1, \sigma_2 \in \Omega$, $\sigma_1\perp^Q\sigma_2$. By checking the order, we know $\sigma _1\sigma _2 \in \Omega$.

If $\supp(\sigma_1)=\langle L_1, \cdots,
L_{k_1}\rangle$ and $\supp(\sigma_2)=\langle L_{k_1+1}, \cdots,
L_{k_1+k_2}\rangle$, so that
	 $$\sigma _1 (z)=\tau _1(L_1(z), \cdots,
L_{k_1}(z)),$$
$$\sigma _2 (z)=\tau _2(L_{k_1+1}(z),
\cdots, L_{k_1+k_2}(z)) $$
for some meromorphic germs of symbols $\tau_i(w_i)$, then the product reads
$$\sigma(z)= \tau _1(L_1(z), \cdots,
L_{k_1}(z)) \, \tau _2(L_{k_1+1}(z),
\cdots, L_{k_1+k_2}(z))= \tau _1(\phi_1(z)) \, \tau _2(\phi_2(z)),$$
where we have set $\phi_1(z):=(L_1(z), \cdots,
L_{k_1}(z))$ and $\phi_2(z):=(L_{k_1+1}(z), \cdots,
L_{k_1+k_2}(z))$.  This shows that $\supp(\sigma_1\, \sigma_2)\subset \supp(\sigma_1)\oplus \supp(\sigma_2)$.

So
\[\left(\supp(\sigma_1)\perp^Q \supp(\sigma_3)\wedge \supp(\sigma_2)\perp^Q \supp(\sigma_3)\right)\Longrightarrow \left( \supp(\sigma_1\, \sigma_2)\perp^Q \supp(\sigma_3)\right),\]
from which it follows that $\left(\Omega, \perp^Q, m_{\Omega}\right)$ is a \loc algebra.
Moreover, $\Omega$ is clearly commutative and has unit given by the constant function.
	\end{proof}

\subsection{Finite part at infinity on $\Omega$} \label{subsec:ev_at_infty} We can now define a finite part at infinity map on $\calm \cals (\C ^\infty)$ by using the  finite part at infinity map on $\cals _{\rm ph}(\R _{\ge 0})$. For a generator of the form
$$\sigma (z)\sim \sum_{j=0}^{\infty } a_{ j}(z)\,x^{\alpha (z)-j},
$$
on $U$, $\alpha (z)=L(z)+\alpha _0$, if $L\in \call$, we set it to be
$$\sigma (z)\mapsto \underset{+\infty}{\rm fp}(\sigma (z))
$$
on $U\setminus \{\alpha(z)+j=0, j \in \Z _{\ge 0}\}$, if $L\not =0$;

and
$$\sigma (z) \mapsto \underset{+\infty}{\rm fp}(\sigma (z))
$$
on $U$, if $L=0$.

\begin{lem}\label{lem:fpOmega} This defines a linear map
 $$\underset{+\infty}{\rm fp}: \calm \cals (\C ^\infty )\to \calm (\C ^\infty),$$
 on $\Omega$, it sends simple meromorphic germs to zero except on polynomials $ P(z)=\sum_{j=0}^k a_j(z)\, x^j$  for which $\underset{+\infty}{\rm fp} (P(z))= a_0(z)$.
\end{lem}
\begin{proof} By linearity, we only need to check simple meromorphic germs of symbols. For a simple meromorphic germ of symbols with order $\alpha (z)=L(z)+\alpha _0$, $L\not =0$, represented by
$$\sigma(z)\sim \sum_{j=0}^{\infty } a_{ j}(z)\,x^{\alpha (z)-j},$$
on $U\setminus X$ with $U$ contained in some open ball centered at zero, and $X=\{L_1=\cdots=L_n=0\}$, $a_j(z), j\in \Z _{\ge 0}$ holomorphic on $U\setminus X$.
By definition, on $U\setminus \{L_1=\cdots=L_n=\alpha -j=0, j\in \Z _{\ge 0}\}$,
$$ \underset{+\infty}{\rm fp} \sigma=0.
$$
By our choice of $U$, $U$ and $\{\alpha -j=0\}$ have no intersection when $j$ is big enough, so $U\setminus \{L_1=\cdots=L_n=\alpha -j=0, j\in \Z _{\ge 0}\}$ is $U\setminus {\rm finite \ hyperplanes} $. So in this case $ \underset{+\infty}{\rm fp} \sigma$ extends to the zero function on $U$, thus the 0 germ.

For a simple meromorphic germ of symbols with order $\alpha (z)=\alpha _0$, represented by
$$\sigma(z)\sim \sum_{j=0}^{\infty } a_{ j}(z)\,x^{\alpha _0-j},$$
on $U\setminus X$ with $U$ contained in some open ball centered at zero, and $X=\{L_1=\cdots=L_n=0\}$, $a_j(z), j\in \Z _{\ge 0}$ holomorphic on $U\setminus X$, and
$$L_1^{s_1}\cdots L_n^{s_n}\sigma =\tau(z)\sim \sum_{j=0}^{\infty } b_{ j}(z)\,x^{\alpha (z)-j}
$$
on $U\setminus X$ with $X=\{L_1=\cdots=L_n=0\}$, and $b_j(z), j\in \Z _{\ge 0}$ holomorphic on $U$. By the uniqueness of the asymptotic expansion, $L_1^{s_1}\cdots L_n^{s_n}a_j=b_j$, so $a_j \in \calm (\C ^\infty)$.

By definition, on $U\setminus \{L_1=\cdots=L_n=0\}$,
$$ \underset{+\infty}{\rm fp} \sigma=\sum a_j\delta _{\alpha _0-j,0}.
$$
Since this is only one term at most, $\underset{+\infty}{\rm fp} \sigma \in \calm (\C ^\infty)$. We have the conclusion.
\end{proof}
We know that $\underset{+\infty}{\rm fp} $ is a partial character on $\Sigma (\R_{\geq 0})$  (see Proposition \ref{prop:fpprod}).
By Proposition \mref {prop: LocCoeff}, we have
\begin{prop} \label{prop:fpOmega}
The finite part at infinity
\begin{eqnarray*}
 \underset{+\infty} {\rm fp}: \Omega&\longrightarrow& {\mathcal M}(\C^\infty),\\
 \sigma(z)&\longmapsto& \underset{+\infty} {\rm fp}(\sigma)(z)
\end{eqnarray*}
	is a morphism of locality algebras i.e., for any two independent germs of symbols $\sigma_1$ and $ \sigma_2 $ in $\Omega$, then
	\begin{equation}  \label{eq:fpOmega}
	  \underset{+\infty}{\rm fp} \left(\sigma_1 \, \sigma_2 \right) = \underset{+\infty}{\rm fp} ( \sigma_1 ))   \, \underset{+\infty}{\rm fp}( \sigma_2 ).
	  \end{equation}
	\end{prop}

 \begin{proof}In view of the linearity  of the map $\underset{+\infty}{\rm fp} $ and the fact that it only affects the $x$-variable, it is sufficient to consider products of simple holomorphic germs.

 Let $\sigma_1$ and $ \sigma_2 $ be two independent  simple holomorphic germs of symbols in $\Omega$ with respective orders $\alpha_1$ and $\alpha_2$ respectively. If for $i=1,2$,
 $$\sigma_i(z)\sim \sum_{n_i=0}^\infty a_{i n_i}(z)\, x^{\alpha(z)-n_i},$$
it follows from Proposition \ref{prop:orderprod} that  $\alpha_1\perp^Q\alpha_2$ and $ a_{1 n_1}\perp^Q a_{2 n_2}$ for any $(n_1, n_2)\in \Z_{\geq 0}^2$.

\begin{itemize}\item Assume that $\alpha_1$ and $\alpha_2$  are affine nonconstant respectively. Consequently,  the holomorphic  map $\alpha: (z_1, z_2)\longrightarrow   \alpha_1(z_1)+ \alpha_2(z_2)$  corresponding to the order of $\sigma_1\, \sigma_2$ is also affine nonconstant.
 By Lemma \ref{lem:fpOmega}, we know that $ z_1\longmapsto \underset{+\infty}{\rm fp}(\sigma_1)(z_1)$, $ z_2\longmapsto \underset{+\infty}{\rm fp}(\sigma_2)(z_2)$ and $(z_1, z_2)\longmapsto \underset{+\infty}{\rm fp}(\sigma_1\, \sigma_2)(z_1, z_2)$ all vanish  as meromorphic functions   respectively. It follows  that Eq.~(\ref{eq:fpOmega}) holds in that case.
  \item We now consider the case when one of the two symbols is polynomial so, let us assume that $\alpha_1$ is affine nonconstant in $z_1$ and $\sigma_2$ is a polynomial of degree $d_2$. In that case $\sigma_1(z_1)\, \sigma_2(z_2)$ is a simple holomorphic germ of order $\alpha_1(z_1)+d_2$ which is also affine nonconstant. The rest of the  above reasoning goes through in a similar manner.
  \item If the two symbols are polynomials say $\sigma_i(z)= \sum_{j_i=0}^{d_i} a^i_{j_i}(z)\, x^{d_i}$, so is their product $\sigma_1(z_1)\, \sigma_2(z_2)= \sum_{k=1}^{d_1+d_2} c_k(z_1, z_2)\, x^{d_1+d_2}$ with $c_k(z_1, z_2)= \sum_{j_1+j_2=k}  a^1_{j_1}(z_1)\, a^2_{j_2}(z_2)$ and we have \[\underset{+\infty}{\rm fp} \left(\sigma_1(z_1)\, \sigma_2(z_2)\right)=c_0(z_1, z_2)= a^1_{0}(z_1)\, a^2_{0}(z_2)= \underset{+\infty}{\rm fp} ( \sigma_1 ))   \, \underset{+\infty}{\rm fp}( \sigma_2 ).\]
\end{itemize}
\end{proof}

\begin{rk}
 Notice that Eq.~\eqref{eq:fpOmega} does not hold for any germs of symbols. For non-independent germs of symbols, the argument of Remark \ref{rk:cancelling orders}, (\ref{rk:cancelling orders i}) provides a counter-example.
\end{rk}

An easy observation:

\begin {prop} The finite part at infinity map equips $\calm (\C ^\infty)$ a $(\Omega, \perp ^Q)$-operated locality algebra structure:
$$\Omega \times \calm (\C ^\infty)\to \calm (\C ^\infty)
$$
$$(\sigma, f)\mapsto f \underset{+\infty}{\rm fp}( \sigma ) .
$$
\end{prop}

	\subsection{\Loc Rota-Baxter operators on $\Omega$}
	
	\begin{thm} \label{thm:RB_Omega}
	\begin{enumerate}
	\item The integration map ${\mathfrak I}  $   extends to  a  map on $\Omega$ and  the triple
			$\left(\Omega, \perp^Q, m_{\Omega},{\mathfrak I}\right)$ is a \loc Rota-Baxter  commutative  algebra of weight zero, and we have ${\rm Id}\perp^Q {\mathfrak I}  $.
			\item The summation  maps ${\mathfrak S}_{\pm 1}  $   extend to  a  map on $\Omega$, the triple
			$\left(\Omega, \perp^Q, m_{\Omega},{\mathfrak S}_{\pm 1}\right)$ is a \loc Rota-Baxter  commutative  algebra of weight $\mp 1$, and we have ${\rm Id}\perp^Q {\mathfrak S}_{\pm 1} $.
			\end{enumerate}
	\end{thm}
	\begin{proof} Since $\Omega$ consists of sums of simple meromorphic germs of symbols with nonconstant affine order and polynomials with meromorphic coefficients and  since the integration map  clearly stabilises the algebra of such polynomials,  it suffices to consider the integration map on  simple meromorphic germs of symbols with nonconstant affine order and it is sufficient to consider  a holomorphic germ of symbols with nonconstant affine order.

Let $\rho(z)$ be such a holomorphic family of order $\alpha(z)$,
 $$\rho(z)\sim  \sum_{j=0}^\infty a_j(z)\, x^{\alpha (z)-j}.$$
 \begin{itemize}
	\item   We want to show that
 ${\mathfrak I}(\rho(z)) \in \Omega$.
 For  $z\notin \alpha^{-1}\left(   \Z\right)$, the symbol   $\rho(z)$  whose order $\alpha(z)$ lies in $\C\setminus \Z$,  is an element of $ \Sigma (\R_{\geq 0})$ so that  we can implement the explicit expression  (\ref{eq:Isigma})  of the integration map     which yields
\begin{eqnarray}\label{eq:Irho}
 {\mathfrak I}(\rho(z))(x) & =&\sum_{j=0}^{N-1}a_j(z)\,\left(\int_0^1 \chi(y)\, y^{\alpha (z)-j}
(y)\, dy+ \frac{x^{\alpha (z)-j+1}}{\alpha (z)-j+1} - \frac{1}{\alpha (z)-j+1}\right)\nonumber \\
                                &+ &   {\int_0^x\tau_{(N)}^\chi(z)(y)\, dy.}\noindent
\end{eqnarray}
The r.h.s defines a  meromorphic germ of polyhomogeneous symbols of order $\alpha (z)+1$     whose poles lie in  $\alpha^{-1}([-1, +\infty[\cap \Z)$ and thus defines  an element  $ {\mathfrak I}(\rho(z)) $ in $\Omega$.

A similar procedure using the explicit expression 	 (\ref{defi_barP_barQ}) shows the corresponding assertion
for ${\mathfrak S}_{\pm 1}(\rho (z))$.
\item The fact that  ${\mathfrak I}$ is a locality map, namely \[ \sigma_1 \perp^Q\sigma_2\Longrightarrow {\mathfrak I}(\sigma_1)\perp^Q{\mathfrak I}(\sigma_2),\] follows from combining Eq.~(\ref{eq:Irho}) with  the first item of Proposition \ref{prop:orderprod}. To deduce the locality of ${\mathfrak S}_{\pm 1}$ from that of  ${\mathfrak I}$,  it is useful to observe that the derivatives of any two  independent holomorphic germs of symbols are also independent \[\rho_1\perp^Q \rho_2\Longrightarrow \partial_x^{(k)}\rho_1\perp^Q\partial_x^{(k)},\quad \forall k\in \N.\]  It is then easy to deduce from the locality of the map ${\mathfrak I}$ and the Euler-Maclaurin formula,  that ${\mathfrak S}_{\pm 1}$ are locality maps.

\item The locality $\lambda$-Rota-Baxter property  of the maps $ {\mathfrak I}$ and ${\mathfrak S}_{\pm 1}$ for $\lambda=0$ and $\lambda=\pm 1$ respectively, is a consequence of the corresponding known usual Rota-Baxter properties.
\end{itemize}
	\end{proof}
	Since the   operators ${\mathfrak S}_0={\mathcal I}$ and ${\mathcal S}_{\pm 1}$  stabilise $\Omega$, we can compose them on the  left with the finite part at infinity investigated in Proposition \ref{prop:fpOmega}.
	 \begin{coro}  \begin{enumerate}
	\item The cut-off integral $\cutoffint_0^\infty:=\underset{+\infty} {\rm fp}\circ{\mathfrak I}$ defines a linear map   \[ \cutoffint_0^\infty: \Omega\longrightarrow {\mathcal M}(\C^\infty)\] compatible with the filtration on the source and target space.
			\item The cut-off sum $\cutoffsum_0^\infty:=\underset{+\infty}{\rm fp}\circ{\mathfrak S}_{\pm 1}$ defines a linear map   \[ \cutoffsum_0^\infty: \Omega\longrightarrow {\mathcal M}(\C^\infty)\]  compatible with the filtration on the source and target space.
			\end{enumerate}
	\end{coro}
\vfill \eject \noindent

\part{Branched zeta values}

We  bring together the algebraic and analytic aspects of this paper to define and   renormalise branched zeta values. These   are higher zeta functions which   generalise the usual multiple zeta values; a
detailed study of convergent branched zeta values was carried out in  \cite{M}.

\section{Branched zeta functions}

{For the rest of the paper, we will take $\bfk$ to be $\C$. }

\subsection{Branching the Rota-Baxter operators ${\mathfrak S}_\lambda$}
Let $\Omega$  be the algebra of admissible meromorphic germs of symbols defined in Subsection \ref{subsec:ev_at_infty} equipped with the independence relation $\perp ^Q$ for
the canonical inner product $Q$,  which we will often  write it as  $\top_\Omega$. With the notations of Definition \ref{defn:prop_dec_forests}, we consider the \loc algebra $\mathcal{F}_{\Omega,\top_\Omega}$ of properly $\Omega$-decorated forests,
which we refer to as the $(\Omega, \top_\Omega)-$\loc algebra of forests properly decorated by meromorphic
	multivariate germs of symbols.  A properly $\Omega$-decorated forest will be denoted by $\F:= (F, \sigma_F)$.

Let $\lambda\in \{\pm 1\}$.    {In view of the properties proved in Theorem \ref{thm:RB_Omega} of the  maps $ {\mathfrak S}_\lambda$  introduced in Definition~\ref{defn:SIlambda} we can apply to them Corollary \ref{coro:existencemapbranching}} and  build the corresponding branched map
\begin{equation}\label{eq:BranchedSlambda}\widehat{{\mathfrak S}_\lambda} : {\C} \mathcal{F}_{\Omega,\top_\Omega}\to \Omega,
\end{equation}
which {is} a morphism of \loc unital algebras.
\begin{prop}\label{pp:barPF}
	
	For  $\lambda\in \{\pm 1\}$,  the map
		\begin{eqnarray*} \mathcal{Z}^\lambda: \C {\mathcal F}_{\Omega, \top_\Omega}&\longrightarrow& \mathcal{M} (\C^{\infty})\\
		  \F:=(F, \sigma_F)&\longmapsto & \mathcal{Z}_{\F}^{\lambda} :={\underset{+\infty}{\rm fp}\left( \widehat{{\mathfrak S}_\lambda}(\F)\right)}
		\end{eqnarray*}
 is a morphism of \loc algebras.
\end{prop}
\begin{proof}
		Combining the facts that  $\widehat{{\mathfrak S}_\lambda}$  and  the finite part map $\underset{{+\infty}}{\rm fp}$
		are local morphisms of \loc algebras  (by Corollary \ref{coro:existencemapbranching} and Proposition \ref{prop:fpOmega} respectively) yields that  the composition
		$\F\mapsto   \mathcal{Z}_{\F}^\lambda$ {gives} a morphism of \loc algebras.
\end{proof}

The subsequent properties are  a straightforward consequence of the fact that $\mathcal{Z}^\lambda$ is a morphism of  \loc algebras and  $\widehat{{\mathfrak S}_\lambda}$ a morphism of \loc  operated  algebras.

\begin{coro} \label{sep_forest}
	The following identities of meromorphic functions hold:
	\begin{enumerate}
		\item For mutually independent properly decorated $\Omega$-forests $\F_1\top_\Omega\F_2$, we have
		 $$\mathcal{Z}^\lambda_{\F_1\F_2}=\mathcal{Z}^\lambda_{\F_1}\cdot \mathcal{Z}_{\F_2}^\lambda.$$
		 i.e. $\mathcal{Z}^\lambda$ is a \loc character.
		\item For a properly decorated $\Omega$-forest $\F=(F, \sigma_F)$ and $\sigma\perp_\Omega \sigma_F(v)$ for all $v\in {\mathcal V}(F)$, thus making $B_+^{\sigma}(\F)$ a properly decorated tree, we have
		$$\mathcal{Z}^\lambda_{ B_+^{\sigma}(\F)}=  {\underset{+\infty}{\rm fp}\left({\mathfrak S}_\lambda\left(\sigma \, \widehat{{\mathfrak S}_\lambda}(\F)\right)\right)}.$$
	\end{enumerate}
\end{coro}
\begin{proof} The first property is a direct consequence of the \loc morphism  property and of Proposition \ref{prop:fpprod}.
	The last property follows from the fact that $\widehat{{\mathfrak S}_\lambda}$ is a morphism of operated \loc algebras for the operation defined in Lemma \ref{lem:Omegaphi}:
	\[\widehat {{\mathfrak S}_\lambda}( B_+^{\sigma}(\F))=  {\mathfrak S}_\lambda\left(\sigma\, \widehat{{\mathfrak S}_\lambda}({\F})\right).\]
\end{proof}

\subsection{Branched zeta functions}
\label{subs:trees_dec_symb}
We next set $E=\Z_{\geq1}\times\C$ endowed with the  independence relation $(k, s)\,\top_E \, (\ell,t):\Longleftrightarrow k\neq \ell$.

\begin{rk}
 As it will become  clear  from the following Lemma, $\Z_{\geq1}$ serves as a label set on vertices to attach a specific element of $\Omega$ to each
 vertex, while $\C$   hosts the ``weights''   of the branched zeta functions yet to be defined.
\end{rk}

Let $\chi(x)$ be an excision function which is identically 1 on $[1,\infty)$.

\begin{defn} \label{defn:calr}
 Let $\calr _\chi:\Z_{\geq1}\times\C\longrightarrow \Omega$ be the map defined by
 \begin{equation*}
 \calr_\chi (\ell,s) := \left(x \mapsto\chi(x)x^{s-z_\ell}\right)
 \end{equation*}
 with $(z_1,z_2,\cdots)$ the canonical coordinates of $\C^\infty$.
\end{defn}

    On the grounds of Theorem \ref{thm:liftedphi}, the \loc map $\calr _\chi :(E=\Z_{\geq1} \times\C,\top_E) \rightarrow (\Omega,\top_\Omega)$  lifts to a  morphism  \[\calr _\chi ^\sharp: { \C\,\calf_{E,\top_E} \rightarrow  \C\,\calf_{\Omega,\top_\Omega}}\]   of \loc algebras.

Given  $\lambda\in \{ \pm 1\}$, {composing the two morphisms ${\mathcal R}_\chi ^\sharp $ and ${\mathcal Z}^\lambda$ of \loc algebras, we obtain}
 \begin{propdefn}\label{prop:zeta}
 	The map
 		\begin{eqnarray*}
 			\zeta _\chi ^{{\rm reg},\lambda}: {\mathcal F}_{E, \top_E}&\longrightarrow& \mathcal{M} (\C^\infty),\\
 			\F&\longmapsto & \mathcal{Z}_{{\mathcal R}_\chi ^\sharp(\F)}^\lambda=(\mathcal{Z}^\lambda \circ \mathcal{R}_\chi ^\sharp)(\F),
 		\end{eqnarray*}
 		defines a morphism of \loc algebras, which to a properly $\Z_{\geq1} \times\C$-decorated tree $\F=(F, d_F)$, assigns a multivariate meromorphic germ at zero, denoted  by  $\zeta^{{\rm reg},\lambda}_{\chi}(\F) $ and  which we call  {\bf regularised branched zeta} functions. Extending the conventions commonly used for multiple zeta functions, when $\lambda=-1$, we drop the upper index writing $\zeta^{{\rm reg} }_{\chi}$ and when  $\lambda=1$ we write $\zeta^{{\rm reg},\star}_{\chi}$.
 \end{propdefn}

The subsequent statement follows on inspection of the concrete  formula.

\begin {prop} The map $\zeta^{{\rm reg},\lambda}_{\chi}$ does not depends on the choice of excision function $\chi$ as long as  it is identically $1$ on $[1,\infty)$. Therefore we will drop the subindex $\chi$ from now on.
\end{prop}

\subsection{Renormalised branched zeta values}
We choose the inner product on $\C^\infty$ to be the canonical inner product.

  	\begin{defn}
  	  Let $\lambda\in\{\pm 1\}$ and let $\F=(F,\vec k)$ be a $\Z_{\geq1}\times\C$-decorated  forest. The {\bf renormalised branched zeta value} (or renormalised BZV) associated to the decorated tree $\F$ is defined  as
  			\begin{equation*}
  			\zeta^{{\rm ren},\lambda}(\F):=\pi_+\circ \zeta ^{{\rm reg},\lambda}(\F)|_{\vec z=\vec 0}= {\rm ev}_0\circ\pi_+\circ \zeta ^{{\rm reg},\lambda}(\F),
  			\end{equation*}  where ${\rm ev}_0$ is the evaluation at zero and $\pi_+$ is the projection operator defined in \cite{GPZ3} associated to the canonical inner product $Q$, making   the subsequent
  			diagramme    commutative.
  $$ \xymatrix{
&& \C\, \calf_{\Z_{\geq1} \times \C,\top} \ar[ddll]_{\calr^\sharp} \ar[ddddrr]^{\zeta^{{\rm reg},\lambda}} \ar[rr]^{\zeta^{{\rm ren},\lambda}} && \C && \\
&&&&&& \\
\C\,\calf_{\Omega,\perp_\Omega} \ar[ddrr]_{\widehat{P}_\lambda} \ar[rrrrdd]^{{\mathcal Z}^\lambda} &&&&&& \calm_+(\C^\infty) \ar[uull]_{ev_0} \\
&&&&&&\\
&& \Omega \ar[rr]_{{\underset{+\infty}{\rm fp}}} && \calm(\C^\infty) \ar[uurr]^{\pi_+} &&
}
$$
  	\end{defn}
The following theorem  ensures the multiplicativity of the regularised branched zeta functions on mutually independent elements.
  	\begin{thm}\label{thm:Zlocmorph}
  	 The map $\zeta^{{\rm ren},\lambda}:\C\,\calf_{\Z_{\geq1} \times \C ,\top}\longrightarrow\C$ is a \loc algebra homomorphism.
  	\end{thm}
  	\begin{proof}
  	 The map $\zeta^{{\rm ren},\lambda}: \F\longmapsto {\rm ev}_0\circ\pi_+\circ \zeta ^{{\rm reg},\lambda}(\F )$ is a morphism on the \loc algebra $\C\,\calf_{\Omega,\top}$ as the composition of \loc morphisms of
  	 \loc algebras, namely  ${\rm ev}_0:  \calm_+(\C^{\infty}) \longrightarrow  \C$ is clearly a \loc character, $\F\longmapsto \zeta ^{{\rm reg},\lambda}(\F): \C\,\calf_{\Omega,\top _Q}\longrightarrow \calm(\C^{\infty})$ is a
  	 \loc morphism  by Proposition \ref{prop:zeta} and $\pi_+$ is a \loc morphism by \cite[Example 3.9]{CGPZ1}.
  	\end{proof}

By a similar argument to {the one used in \cite {CGPZ2}}, we have the following statement.
\begin {prop}  For properly $\Z_{\geq1} \times \C$-decorated forest $\F$, $\zeta^{{\rm ren},\lambda}(\F)$ does not depend on the $\Z _{\ge 1}$ decorations.
\end{prop}

\subsection{Branched zeta-functions in terms of multiple zeta functions  }

We now relate our constructions to that of multiple zeta functions carried out in \cite{MP}, identifying a proper word with a properly decorated ladder tree.

  \begin{prop} \label{pp:QcombMZV}
 Let $\F=(F,\vec{k})$ be a ladder tree $\F$ with $k$ vertices decorated from bottom to top  by $(\ell,s_\ell)\in \Z _{\ge 1}\times \C$.
  \begin{enumerate}
  \item
If $\Re(s_1)>1$ and $\Re(s_i)\geq 1, i=2,\cdots,k$, then
\[\zeta^{\rm reg}(\F)=\zeta(s_1-z_1, \cdots, s_k-z_k)=\sum_{{1\leq n_k<\cdots < n_1}} n_1^{-s_1+z_1}\cdots n_k^{-s_k+z_k},\]
for the multiple zeta function.
Similarly,
\[ \zeta^{{\rm reg}, \star}=\zeta^\star(s_1-z_1, \cdots, s_k-z_k):=\sum_{{1\leq n_k\leq \cdots \leq n_1}} n_1^{-s_1+z_1}\cdots n_k^{-s_k+z_k},\]
for the multiple star zeta function.
\item In general, $\zeta^{\rm reg}(\F)$ and $\zeta^{{\rm reg},\star}(\F)$ coincide with    the meromorphic germs
    \[\cutoffsum_{1\leq n_k<n_{k-1}<\cdots <n_1}n_1^{s_1-z_1}\cdots n_k^{s_k-z_k}, \quad {\rm resp.}\quad  \cutoffsum_{1\leq n_k\leq n_{k-1}\leq \cdots \leq n_1}n_1^{s_1-z_1}\cdots n_k^{s_k-z_k}\]
    considered in \cite[Theorem 9]{MP}.
 \end{enumerate}
 \end{prop}

 \begin{proof}
 \begin{enumerate}
 \item This follows from the fact that $\widehat P$ restricted to words coincides with $\widehat P^ {\calw}$ (see  (Eq.~\eqref{eq:hatPonW}) applied to the Rota-Baxter operator ${\mathfrak S}_\lambda$.
\item {Since the cut-off Chen sums considered in  \cite[Theorem 9]{MP}  were built by means of the composition $\underset{{+\infty}}{\rm fp}\circ\widehat {{\mathfrak S}_\lambda}^{\mathcal W}\vert_{{\mathcal W}_{\Omega, \top_\Omega}}$,  the statement   follows from the fact that $ \widehat {{\mathfrak S}_\lambda}\vert{{\mathcal W}_{\Omega, \top_\Omega}}=\widehat {{\mathfrak S}_\lambda}^{\mathcal W}\vert_{{\mathcal W}_{\Omega, \top_\Omega}}$.}
 \end{enumerate}
 \end{proof}

\begin{thm} \label{thm:rational}
Let $\F=(F,\vec{k})\in \calf_{\Z,\top}$ be any properly $\Z_{\geq 1} \times \C$-decorated forest.
\begin{enumerate}
\item $\zeta^{{\rm ren}, \lambda}(\F )$ is a $\Q$-linear combination  of renormalised (ordinary) multiple zeta functions $\zeta$ (resp. $\zeta^\star$) if  $\lambda=1$ (resp. if $\lambda=-1$).
\item
   Provided  the inner product $Q$ is rational,  the renormalised branched zeta values    $\zeta^{{\rm ren}}_{\lambda} (\F)$ associated to any decorated tree $\F=(F, \vec k)$ are  rational
   whenever $s_v\in\Z_{\leq-1}, \forall v\in\calv(\T)$.
   \end{enumerate}
\end{thm}
\begin{proof}
\begin{enumerate}
\item The first statement follows from combining $\widehat {{\mathfrak S}_\lambda}= \widehat {{\mathfrak S}_\lambda}^{\mathcal W}\circ f_\lambda$
 with the fact that ladder trees give rise to multiple zeta functions.    \item
The second statement follows from the first one combined with the fact that renormalized multiple zeta values are rational  {(\cite[Theorem 9]{MP})}.  {Alternatively, one shows by an induction on the number of
vertices of $\F$ that $\widehat {{\mathfrak S}_\lambda}$ has rational coefficients in the sense of \cite{GPZ3}. In view of the rationality  of the inner product  $Q$ the projection  map $\pi_+^Q$  preserves rationality so   that $\zeta ^{{\rm reg},\lambda}(\F)$ has rational coefficients and thus
$\zeta^{{\rm ren}, \lambda}(\F )$ is rational.}
\end{enumerate}
\end{proof}

\end{document}